%% file: gorilla-full-version.tex
\documentclass{article}
\usepackage{amsmath,cases,amsthm}
\usepackage{dirtree}
\usepackage{hyperref}
\usepackage{color}
\usepackage{algorithm}
\usepackage[noend]{algpseudocode}
\usepackage{blindtext}
\usepackage{scalerel}

\def\vecsign#1{\rule[1.388\LMex]{\dimexpr#1-2.5pt}{.36\LMpt}%
  \kern-6.0\LMpt\mathchar"017E}

\usepackage{stackengine,amsmath,amsthm}
\stackMath
\usepackage{graphicx}

\usepackage{float}
\usepackage{caption}
\usepackage{subcaption}
\usepackage{amsfonts}
\theoremstyle{plain}

\usepackage[colorinlistoftodos]{todonotes}
\newcounter{todocounter}

\usepackage[normalem]{ulem}

\newcommand{\msg}[1]{\raisebox{.5pt}{\textcircled{\raisebox{-.9pt} {#1}}}}

\newcommand{\cmdBroadcast}{\ensuremath{ \textit{Broadcast} }\xspace}
\newcommand{\cmdReceive}{\ensuremath{ \textit{Receive} }\xspace}
\newcommand{\cmdDecide}{\ensuremath{ \textit{Decide} }\xspace}
\newcommand{\vala}{\ensuremath{0}\xspace}
\newcommand{\valb}{\ensuremath{1}\xspace}
\newcommand{\varR}{\ensuremath{\mathcal{T}}\xspace}
\newcommand{\varpriority}{\ensuremath{ \textit{priority} }\xspace}
\newcommand{\varpriorityCounter}{\ensuremath{ \textit{uCounter} }\xspace}
\newcommand{\slot}{\ensuremath{ \text{tick} }\xspace}	
	
\newcommand{\GMPLUS}{\ensuremath{ \text{GM+} }\xspace}

\newcommand{\GSand}{{\rm Gorilla}\xspace}

\theoremstyle{plain}

\newtheorem{lemma}{Lemma}

\newtheorem{claim}{Claim}
\newtheorem{observation}{Observation}
\newtheorem{definition}{Definition}

\usepackage{thmtools}	
\usepackage{thm-restate}

\title{Gorilla: Safe Permissionless Byzantine Consensus}

\author{
	Youer Pu\\
	\texttt{Cornell University}
	\and
        Ali Farahbakhsh\\
        \texttt{Cornell University}
        \and
	Lorenzo Alvisi\\
	\texttt{Cornell University}
	\and
	Ittay Eyal\\
	\hspace{10 mm}\texttt{Technion}\hspace{10 mm}
}

\begin{document}

\maketitle

\begin{abstract}
Nakamoto’s consensus protocol works in a permissionless model and tolerates Byzantine failures, but only offers probabilistic agreement. Recently, the Sandglass protocol has shown such weaker guarantees are not a necessary consequence of a permissionless model; yet, Sandglass only tolerates benign failures, and operates in an unconventional partially synchronous model. We present Gorilla Sandglass, the first Byzantine tolerant consensus protocol to guarantee, in the same synchronous model adopted by Nakamoto,  deterministic agreement and termination with probability 1 in a permissionless setting. We prove the correctness of Gorilla by mapping executions that would violate agreement or termination in Gorilla to executions in Sandglass,  where we know such violations are impossible. Establishing termination proves particularly interesting, as the mapping requires reasoning about infinite executions and their probabilities.
\end{abstract}

\input{introduction}
    \input{related}
    \input{model}

\input{protocol}
    \input{correctness}

\input{conclusion}

\bibliography{bibliography}

\appendix
\input{sandglassplus}
\input{proof}

\end{document}

%% file: introduction.tex
    \section{Introduction}

Nakamoto's Bitcoin~\cite{bitcoin} demonstrated that a form of consensus can be reached even if participation is permissionless. 
Nakamoto achieved this by introducing the cryptographic primitive Proof of Work (PoW)~\cite{dwork1992pricing,jakobsson1999proofs} into the common synchronous Byzantine model~\cite{partial}. With PoW, a process can work for a short while and probabilistically succeed in solving a puzzle. 
But Bitcoin only achieves a probabilistic notion of consensus: both safety and liveness fail with negligible probability. 

Lewis-Pye and Roughgarden showed that deterministic and permissionless consensus cannot be achieved in a synchronous network in the presence of Byzantine failures~\cite{Lewis-Pye}. 
Nonetheless, previous work~(\S\ref{sec:Related}) has achieved deterministic safety and termination with probability~1 under different models. 
Sandglass~\cite{Sandglass} assumes a benign model with a non-standard hybrid synchrony model.
Momose et al.~\cite{momose} guarantee termination only if the set of processes stabilizes. 
Malkhi et al.~\cite{malkhi}, while leveraging either authenticated channels or digital signatures, propose a solution whose correctness depends on Byzantine nodes comprising fewer than a third of the nodes in the system.

The question is whether it is possible to achieve deterministic safety and termination with probability~1 without limiting either Byzantine behavior or how nodes join and leave, and without relying on authentication. 

We answer this question in the affirmative for a synchronous model~(\S\ref{sec:model}) with Byzantine failures.
We present \emph{Gorilla Sandglass} (or simply \emph{Gorilla})~(\S\ref{sec:protocol}), a consensus protocol that guarantees deterministic safety and termination with probability~1 in this standard model, which we dub \emph{GM} (for \emph{Gorilla Model}).
Gorilla relies on a form of PoW: \emph{Verifiable Delay Functions} (\emph{VDFs})~\cite{VDF}.
We consider an ideal VDF~\cite{sprints} that proves a process waited for a certain amount of time and cannot be amortized.
The key difference between a VDF and Nakamoto's PoW is that multiple processes can calculate multiple VDFs concurrently, but cannot, by coordinating, reduce the time to calculate a single VDF.
The crux of the protocol is simple. 
The protocol proceeds in steps. 
In each step, all (correct) nodes collect VDF solutions from their peers and build new VDFs based on those.
Intuitively, correct processes, which are the majority, accrue solutions faster than Byzantine nodes, and progress through the asynchronous rounds of the protocol faster.
Eventually, the round inhabited by correct nodes is so far ahead of that   occupied by Byzantine nodes that, no longer subject to Byzantine influence,  correct nodes can safely decide.

Gorilla Sandglass adopts the general approach of Sandglass~\cite{Sandglass}, in the sense that puzzle results are accrued, with each puzzle built on its predecessors. 
In Sandglass participants are benign and they send, in each step, a message built on previously received messages.
In Gorilla, however, the Byzantine adversary is not limited to acting on step boundaries or communicating at particular times.
Surprisingly, Gorilla's correctness can be reduced to the correctness of a variation of Sandglass.
We perform this reduction in two steps~(\S\ref{sec:correctness}). 

We first show that, for every execution of Gorilla in GM, there is a matching execution where the Byzantine processes adhere to step boundaries, in a model we call \emph{GM+}. 
In the mapped execution, Byzantine processes only start calculating their VDF at the beginning of a step and only send messages at the end of a step. 
GM+ is a purely theoretical device, as it allows operations that cannot be implemented by actual cryptographic primitives. In particular, it allows Byzantine processes to start calculating a VDF in a step~$s$ building on any VDF computed by other  Byzantine nodes that will be completed {\em by the end} of~$s$, rather than by the start $s$, as allowed by GM (and actually feasible in reality). Nonetheless, GM+ serves as a crucial stepping stone towards proving Gorilla's correctness.
% We show that if there exists, with positive probability, an execution where the adversary violates correctness in GM, then, with positive probability, there is an execution that violates correctness in GM+. 

Next, we show that, given an execution in GM+ that violates correctness, there exists a corresponding execution of Sandglass in a model we call SM+.  
The SM+ model is similar to that of Sandglass: in both, processes are benign and propagation time is bounded for messages among correct processes and unbounded for  messages to and from so-called \emph{defective nodes}. 
But unlike Sandglass, in SM+ a message from a defective node can reference another message generated by another defective node during the same step (similar to how GM+ allows Byzantine nodes to calculate a VDF that builds on VDFs calculated by other Byzantine nodes in the same step).  

Together, this pair of reduction steps establishes that if an execution of Gorilla in~GM violates correctness with positive probability, then so does an execution of Sandglass in~SM+. To conclude Gorilla's proof of correctness, all that is left to show is that Sandglass retains deterministic safety and termination with probability~1 in the SM+ model: fortunately, the correctness proof of Sandglass~\cite{Sandglass} works almost without change~(\S\ref{sec:sm}) in SM+.
Thus, a violation of correctness in Gorilla results in a contradiction, and therefore, Gorilla is correct.

Gorilla demonstrates that it is {\em possible} to achieve deterministic safety and liveness with probability~1 in a permissionless Byzantine model. Yet, possible does not mean {\em practical}: Gorilla is not, since, like the Sandglass protocol that inspires it, it requires an exponential number of rounds to terminate. By answering the fundamental question of possibility, Gorilla ups the ante: is there a practical solution to deterministically safe  permissionless consensus?

%% file: related.tex
\section{Related Work}
\label{sec:Related}

Lewis-Pye et al.~\cite{Lewis-Pye} have proven that deterministic consensus is impossible in the permissionless setting. Therefore, for at least one of safety and liveness probabilistic guarantees are inevitable. Gorilla concedes little: it manages to keep safety deterministic, and guarantees liveness with probability~1. 

Several protocols~\cite{ouroboros-genesis,prism,snow-white,parallel-chains,algorand,ouroboros,ebb-flow} have embraced Bitcoin's permissionless participation and probabilistic safety.
All rely for correctness on probabilistic mechanisms, which leave open the possibility that  Byzantine nodes may overturn safety or liveness guarantees with positive probability.
Gorilla avoids this peril by basing correctness on the process of accruing a deterministic number of messages.

Few proposals achieve deterministic safety in a permissionless setting~\cite{malkhi,momose, Sandglass}. Momose et al.~\cite{momose} introduce the concept of {\em eventually stable participation}, akin to partial synchrony; it requires that, after an unknown global stabilization time, for each T-wide time interval $[t, t+T]$, at least half of the 
nodes ever awake during the interval are correct and do not leave. Gorilla guarantees progress
without assuming stability in participation.

Pu et al.~\cite{Sandglass} propose Sandglass, which achieves deterministic safety but only in a benign setting. Gorilla extends Sandglass to tolerate Byzantine failures.

Malkhi et al.~\cite{malkhi} let nodes join and leave at any time, as in Gorilla. Unlike Gorilla, however, they must rely on authenticated channels to tolerate fluctuations in the number of adversaries. Further, Byzantine nodes must be fewer than a third of active nodes, while in Gorilla they must be fewer than one half.

Several works have modeled permissionless participation~\cite{wait-free,resource-pools,sleepy}.

Pass et al.~\cite{sleepy} introduce the {\em sleepy participation model}, in which honest nodes are either awake or asleep.
Awake nodes participate in the protocol, while asleep nodes neither participate nor
relay messages.
Byzantine nodes are always awake, but the scheduler can adaptively turn an honest node Byzantine, as long as Byzantine nodes remain a minority of awake nodes. Gorilla 
similarly assumes that correct and Byzantine nodes can join and leave at any time, as long as a majority of active nodes are correct. Unlike the sleepy model, however, Gorilla requires no public key infrastructure, and, unlike sleepy consensus, guarantees deterministic safety. 

Unlike Gorilla, Lewis-Pye et al.~\cite{resource-pools} do not  offer a consensus protocol, but rather focus on introducing {\em resource pools}, an abstraction that aims to capture resources used to establish identity in permissionless systems, {\em e.g.}, computational power through PoW and
fiscal power through Proof of Stake (PoS). 

Aspnes et al.~\cite{wait-free} explore consensus in an asynchronous benign model where an unbounded number of nodes can join and leave, but where at least one node is required to live forever, or until termination. Gorilla instead assumes a synchronous model, tolerates Byzantine failures, and allows any node to join and leave, as long as a majority of active nodes is correct.

Verifiable Delay Functions (VDFs)~\cite{VDF} have been leveraged as a resource against Byzantine adversaries in various works~\cite{posat, POSH, VDP, fairledger}, specifically to defend PoS systems from attacks where participants can go back in time and mine blocks. Gorilla leverages VDFs to rate-limit the ability of Byzantine nodes to create valid messages.

%% file: model.tex
\section{Model} 
\label{sec:model} 

	The system is comprised of an infinite set of nodes~$\{p_1, p_2, \dots\}$. 
    Time progresses in discrete ticks~$0,1, 2, 3, \dots$ In each tick, a subset of the nodes is \emph{active}; the rest are \emph{inactive}. 
   The upper bound on active nodes in any \slot, necessary to the safety of Nakamoto's permissionless consensus~\cite{pass17Analysis}, is~${\mathcal N}$, and there is at least one active node in every \slot.
    Starting from tick 0, every~$K$ ticks are grouped into a step: each step~$i$ consists of ticks~$iK, iK+1, \dots, iK+K-1$.
    	
    %\item[VDF] 
    A Verifiable Delay Function (VDF) is a function whose calculation requires completing a given number of sequential steps. Thus, evaluating a VDF requires the evaluator to spend a certain amount of time in the process. Specifically, we require the evaluation of a single VDF to take~$K$ ticks. 
    %Note that Byzantine nodes may choose not to perform the calculation consecutively. 
    We refer to the intermediate random values that this evaluation produces at the end of each of the $K$ ticks as the {\em units of the VDF evaluation} (or, more succinctly, the \emph{units of the VDF}). We denote the $i$-th unit of evaluating the VDF of some input $\gamma$ by $\textit{vdf}_\gamma^{\,i}$; we denote the final result ({\em i.e.,} $\textit{vdf}_\gamma^{\,K}$) by $\textit{vdf}_\gamma$, or, when there is no ambiguity, by {\em vdf}.
    
    We model the calculation of VDFs with the help of an \emph{oracle} $\Omega$. Nodes use $\Omega$ both to iteratively obtain the units of a VDF and to verify whether a given value is the {\em vdf} of a given input. In particular, 
    $\Omega$ provides the following API:

 \begin{description}
		\item[Get($\gamma, \textit{vdf}_\gamma^{\,i}$):] returns $\textit{vdf}_\gamma^{\,(i+1)}$. 
		By convention, invoking Get($\gamma, \bot$) returns $\textit{vdf}_\gamma^{\,1}$. The oracle remembers how it responded to a Get query -- so that, even though the units of a VDF are random values, identical queries produce identical responses. $\Omega$ accepts at most one call to Get() in any tick from each node.

		\item[Verify($\textit{vdf},\gamma$):]
		returns True iff~$\textit{vdf} = \textit{vdf}_\gamma^{\,K}$.
		$\Omega$ accepts any number of calls to Verify() in any tick from any node.
	\end{description}
	
	If Get($\gamma,\bot$) is called at tick~$t$ and step~$s$, we say the VDF calculation for~$\gamma$ \emph{starts} at tick~$t$ and step~$s$.
	Similarly, the VDF calculation for~$\gamma$ \emph{finishes} at tick~$t$ and step~$s$ if Get($\gamma,\textit{vdf}_\gamma^{\,K-1}$) is called at tick~$t$ and step~$s$.

In each \slot, an active node receives a non-negative number of messages, updates its variables -- potentially including calls to the oracle -- and then communicates with others using a synchronous broadcast network. 
The network allows each active node to \emph{broadcast} and \emph{receive} unauthenticated messages.
 Node~$p_i$ invokes ~$\cmdBroadcast_i(m)$ to broadcast a message~$m$, and receives broadcast messages from other nodes (and itself) by invoking~$\cmdReceive_i$.
The network neither generates nor duplicates messages and ensures that if a node receives a message~$m$ in \slot~$t$, then~$m$ is broadcast in \slot~$(t-1)$.
The network propagation time is negligible compared to a \slot, {\em i.e.}, to the time necessary to calculate a unit of a VDF.
By executing the command~$\cmdReceive_i$, a newly joining node $p_i$ receives all messages broadcast by correct nodes prior to its activation.
Nodes whose network connections with other nodes are asynchronous can be modeled as Byzantine, as Byzantine nodes can deliberately or unintentionally delay messages sent from or to them. Therefore, Gorilla also tolerates asynchrony, as long as the nodes that communicate asynchronously are a minority.

 Correct nodes do not deviate from their specification and constitute a majority of active nodes at each tick.
Correct nodes always join at the beginning of a step and leave when a step ends. 
Hence, a correct node is active from the first to the last tick of a step.
The remaining nodes are {\em Byzantine} and can suffer from arbitrary failures. 
Byzantine nodes can join and leave at any \slot.

All nodes are initialized with a value~$v_i \in \{\vala, \valb\}$ upon joining the system.
An active node~$p_i$ decides by calling  $\cmdDecide_i(v)$ for some value~$v$. 
A protocol solves the consensus problem if it guarantees the following properties~\cite{partial}:
	
    \begin{definition}[Agreement]
		\label{def:agreement}
		If  a correct node decides a value~$v$, then no correct node decides a value other than $v$.
	\end{definition}
	
	\begin{definition}[Validity]
		If all nodes that ever join the system have initial value~$v$ and there are no Byzantine nodes, then no correct node decides $v' \neq v$. 
	\end{definition}
	
	\begin{definition}[Termination]
		Every correct node that remains active eventually decides.
	\end{definition}

%% file: protocol.tex
\section{Gorilla}
\label{sec:protocol}

Gorilla borrows its general structure from Sandglass (see Algorithm~\ref{protocol:gorilla})~\cite{Sandglass}. 
Executions proceed in asynchronous rounds (even though, unlike Sandglass, Gorilla assumes a standard synchronous model of communication between all nodes).  Upon receiving a threshold of valid messages for the current round, nodes progress to the next round; if all the messages received by a correct node propose the same value~$v$ for sufficiently many consecutive rounds, the node decides~$v$. The number of active nodes is bounded by~$\mathcal{N}$ but otherwise unknown. Within this bound, it can fluctuate arbitrarily, but both safety and liveness depend on the correctness of a majority of nodes.

The key aspects of the protocol can be summarized as follows:

\begin{description}
\item[Ticks, steps and VDF]
Each valid message must contain a \textit{vdf}.
A correct node takes a full step, \emph{i.e.},~$K$ consecutive ticks, to individually calculate a \textit{vdf}, and at the end of the step sends a valid message that contains the \textit{vdf}. Byzantine nodes may instead share among themselves the work required to finish the~$K$ units of a VDF calculation; even so, it still takes~$K$ distinct ticks for Byzantine nodes to compute a  \textit{vdf}. Requiring valid messages to carry a  \textit{vdf} limits Byzantine nodes to sending messages at the same rate as correct nodes; this ensures that, on average across all steps, the correct majority sends at least one more valid message than the minority of nodes that are Byzantine.

\item[Choosing a threshold]   
A node proceeds to round~$r$ if it receives at least $\varR = \lceil \frac{{\mathcal N}^2}{2} \rceil$ messages for round~$r-1$.
Even though setting such a threshold does not prevent Byzantine nodes from advancing from round to round, it nonetheless gives the correct nodes an edge in the pace of such progress, since they constitute a majority.

\item[Exchanging messages] In each step of the protocol, a node in any round $r$ -- based on the messages it has received so far -- searches for the largest round~$r_{{\mathit max}} \ge r$ for which it has accrued \varR messages. It then broadcasts a message for the next round.
The message includes the node's current proposed value~$v$, the \textit{vdf}, and four other attributes discussed below: the message's \textit{coffer}, a nonce, as well as $v$'s \textit{priority} and \textit{unanimity counter}.

\item[Keeping history] Nodes can join the system at any time.
To help a joining node catch up, every message broadcast by a node $p$ in round $r$ includes a {\em message coffer} that contains: $(i)$ messages from round~$r-1$ received by $p$ to advance to round $r$; $(ii)$ recursively, messages included in those messages' coffers; and $(iii)$ messages received by $p$ for round $r$.

\item[Nonce] By making it possible to distinguish between messages that are generated from the same coffer, nonces allow correct nodes to broadcast multiple valid messages during a round while, at the same time, preventing Byzantine nodes from reusing the same \textit{vdf} to send multiple valid messages based on a given message coffer.

%Gorilla allows nodes to broadcast multiple messages during a round, distinguished by different \textit{nonces} in the messages.

\item [Priority and unanimity counter] If a node~$p$ only receives the value~$v$ from a majority for a sufficient number of consecutive rounds, it decides~$v$.
To guarantee the safety of this decision,~$p$ assigns a {\em priority} to the value~$v$ that it proposes. This priority is incremented once~$v$ is unanimously proposed for a long stretch of consecutive rounds.
To record the length of this stretch, each node computes it upon entering a round~$r$, and includes it as the {\em unanimity counter} in the messages it sends for round~$r$.
If a node collects more than one value in a round~$r$, it chooses the one with the highest priority, and proposes it for round~$r+1$. In case of a tie, it uses~\textit{vdf} as a source of randomness to choose one of the values randomly.
Since~\textit{vdf} is a random number calculated based on the message coffer and a nonce (lines~\ref{gorilla:vdf1:start}-\ref{gorilla:vdf1:end}), a Byzantine node is unable to deliberately pick an input to VDF to deterministically get the desired value.

\item[Message internal consistency and validity]
    A message~$m$ is {\em internally consistent} if the attributes carried by~$m$ can be generated by following Gorilla correctly based on the message coffer carried in~$m$. We denote the~$\textit{vdf}$ in~$m$ by~$\textit{vdf}_m$.

    A message~$m$ is \textit{valid} (and thus  isValid($m$) returns true), if ($i$) $\textit{vdf}_m$ can be verified by the message coffer and the nonce of~$m$; ($ii$)~$m$ is internally consistent; and  ($iii$) for any message~$m'$ in~$m$'s coffer,~$m'$ is also valid.
    Otherwise,~$m$ is invalid.

\end{description}

In addition to demonstrating variable initialization, Algorithm~\ref{protocol:gorilla} presents the algorithm each node~$p_i$ runs at each step.
Each node~$p_i$ starts every step by adding all valid messages, in addition to the messages in their coffers, to the set~$Rec_i$ (lines~\ref{gorilla:forunionReceivedMessages}-\ref{gorilla:unionReceivedMessages}).

Iterating over~$Rec_i$, node~$p_i$ computes the largest round~$r_{max}$ for which it has received at least \varR messages, and updates its current round to~$r_{max} + 1$ (line~\ref{gorilla:enterNewRound}) if the condition in line~\ref{gorilla:roundNumberMax} holds.
 Once in a new round,~$p_i$ does the following:
 ($i$) resets its message coffer~$M$ and adds to it the messages it has received from the previous round -- alongside the messages in {\em their} coffers (lines~\ref{gorilla:resetM}-\ref{gorilla:unionCollectedMessagesEnd}); ($ii$) picks a nonce and calculates a \textit{vdf} based on its coffer and the nonce (lines~\ref{gorilla:vdf1:start}-\ref{gorilla:vdf1:end}); 
($iii$) chooses its proposal value (lines~\ref{gorilla:multiset} -\ref{gorilla:valueUpdateEnd}); it chooses the proposal with the highest priority among the previous round messages in its coffer; in case of a tie, it chooses a random number utilizing the randomness in~$vdf$; ($iv$) determines the priority and the unanimity counter for the messages it will broadcast in the current round (lines~\ref{gorilla:checkValues}-\ref{gorilla:setPriority}); and finally  ($v$) the node decides~$v$ if~$v$'s priority is high enough (lines~\ref{gorilla:decideStart}-\ref{gorilla:decideEnd}).
If~$p_i$ does not enter a new round, it starts to create a message nonetheless: it adds to the message's coffer all messages received for the current round (line~\ref{gorilla:unionSameRoundMsg}), and calculates a~\textit{vdf} with the new message coffer and a different nonce as the input (lines~\ref{gorilla:vdf2:start}-\ref{gorilla:vdf2:end}), so that the message is unique.
Regardless of whether it enters a round or not, $p_i$ ends every step by broadcasting the message it has created (line~\ref{gorilla:broadcast}).

\begin{algorithm}[!htb]
  	\caption{\GSand: Code for node $p_i$. The orange text highlights where Gorilla departs from Sandglass.}
        \label{protocol:gorilla}

  	\begin{algorithmic}[1]
  		\Procedure{Init}{$input_i$} \label{gorilla:Init}
  		\State{~$v_i\gets input_i$;~$priority_i\gets 0$;~$\varpriorityCounter_i \gets 0$;~$r_i=1$;~$M_i = \emptyset$;~$Rec_i = \emptyset$;}
  		\label{gorilla:initialsetup}
  		\EndProcedure
  		
  		\Procedure{step}{} \label{gorilla:step}
  		
  		\ForAll{$m=(\cdot,\cdot,\cdot,\cdot,\cdot,M)$ received by~$p_i$ }\label{gorilla:forunionReceivedMessages}
	  	\color{orange}
            \If{isValid(m)} \label{gorilla:isValid}
            \color{black}
            \State $Rec_i \gets Rec_i \cup \{m\} \cup M $\label{gorilla:unionReceivedMessages}
	  		\EndIf
            
  		\EndFor
  		
  		\If{$\max_{|Rec_i(r)| \ge \varR}(r) \ge r_i$ \label{gorilla:roundNumberMax}}

  		\State $r_i =  \max_{|Rec_i(r)| \ge \varR}(r) + 1$ \label{gorilla:enterNewRound}

  		\State $M_i = \emptyset$ \label{gorilla:resetM}
  		\ForAll{$m=(\cdot, r_i-1,\cdot,\cdot,\cdot,M) \in Rec_i(r_i-1)$ } \label{gorilla:unionCollectedMessagesStart}
  		\State $M_i \gets M_i \cup \{m\}\cup M$\label{gorilla:unionCollectedMessagesEnd}
  		\EndFor		

            \State $M_i \gets M_i \cup Rec_i(r_i)$
            \color{orange}
            \State{$\textit{vdf} \gets \bot$;  $nonce \gets$ a new arbitrary value}
            \label{gorilla:vdf1:start}
            \For{$j:1..k$}
            \State $\textit{vdf} \gets Get((M_i, nonce),vdf)$
            \EndFor\label{gorilla:vdf1:end}
            \color{black}
  		
  		\State Let~$C$ be the multi-set of messages in $M_i(r_i-1)$ with the largest priority. 
  		\label{gorilla:multiset}
  		
  		\If{all messages in $C$ have the same value~$v$} \label{gorilla:valueUpdateStart}
  		\State $ v_i \gets v$
  		\Else
            \color{orange}
                \State $v_i \gets \textit{vdf} \mod  2$ \label{gorilla:vdfValueUpdate}
      
        \label{gorilla:valueUpdateEnd}
            \color{black}
      	\EndIf
  		
  		\If{all messages in $M_i(r_i-1)$ have the same value~$v_i$}
  		\label{gorilla:checkValues}
  		\State{$\varpriorityCounter_i \gets 1 + \min\{\varpriorityCounter | (\cdot,r_i-1,v_i,\cdot,\varpriorityCounter, \cdot) \in M_i(r_i-1)\}$} 
  		\label{gorilla:setPriorityCounter}
  		\Else
  		\State{$\varpriorityCounter_i \gets 0$} \label{gorilla:resetPC}
  		\EndIf
  		
  		\State $priority_i \gets \max (0,\left\lfloor \frac{\varpriorityCounter_i}{\varR} \right\rfloor -5)$ \label{gorilla:setPriority}
  		
  		\If{$\varpriority_i \ge 6\varR+4$} \label{gorilla:decideStart}
  		\State{Decide$_i$($v_i$)} \label{gorilla:decidePoint}	
  		\EndIf	\label{gorilla:decideEnd}	
        
        \Else 	\State $M_i \gets M_i \cup Rec_i(r_i)$
  		\label{gorilla:unionSameRoundMsg}
            \color{orange}
       	\State{$\textit{vdf} \gets \bot$;  $nonce \gets$ a new arbitrary value}\label{gorilla:vdf2:start}
       	\For{$j:1..k$}
       	\State $vdf = Get((M_i, nonce),vdf)$
       	\EndFor\label{gorilla:vdf2:end}
            \color{black}

  		\EndIf
  	\State broadcast~$(r_i,v_i,priority_i, \varpriorityCounter_i, M_i, nonce, \textit{vdf})$ \label{gorilla:broadcast}
  		\EndProcedure
  		
  	\end{algorithmic}
  	
  \end{algorithm}

\subsection{Comparing Sandglass and Gorilla}
    \label{sec:compareGS}

    Gorilla retains the structure of Sandglass, adding the requirement that valid messages must include a \textit{vdf} and a nonce. The differences between the protocols are highlighted in orange in Algorithm~\ref{protocol:gorilla}:
    ($i$) \textit{vdf} is calculated for each message sent (lines~\ref{gorilla:vdf1:start}-\ref{gorilla:vdf1:end},\ref{gorilla:vdf2:start}-\ref{gorilla:vdf2:end}), ($ii$) received messages are checked to see if they are valid (line~\ref{gorilla:isValid}); ($iii$) \textit{vdf} is used as the source of randomness (line~\ref{gorilla:vdfValueUpdate}) where the protocol requires choosing a value randomly. 

    These additions are critical to handling Byzantine faults. Both Gorilla and Sandglass rely on correct (respectively, good) nodes sending the majority of unique messages during an execution. In Sandglass, where defective nodes are benign, this property simply follows from requiring correct nodes to be a majority in each step; not so in Gorilla, where faulty nodes can be Byzantine. Requiring valid message in Gorilla to carry a \textit{vdf} preserves correctness by effectively rate-limiting Byzantine nodes' ability to create valid messages.

    Given their differences in both failure model and timing assumptions, it is perhaps surprising that so little needs to change when moving from Sandglass to Gorilla.  After all, Sandglass assumes a model where failures are benign and a hybrid synchronous model of communication~\cite{cryptoeprint:2022/796}; Gorilla instead assumes a Byzantine failure model, and a synchronous network model (\S\ref{sec:model}). Note, however, that although Sandglass assumes benign failures, its hybrid communication model implicitly accounts for Byzantine nodes strategically choosing the timing for receiving and sending messages to correct nodes: Gorilla can then simply inherit from Sandglass the mechanisms for tolerating such behaviors.

%% file: correctness.tex
	\section{Correctness}\label{sec:correctness}

Despite the similarlity between the Gorilla and Sandglass protocols, proving Gorilla's correctness directly is challenging.
Unlike Sandglass, Byzantine nodes can act between step boundaries, interleave VDF computations instead of producing one VDF (and hence one message) at the time, etc. 
To overcome this complexity, our approach is to leverage as much as possible Sandglass's proof of correctness. 

Our battle plan was to first map executions of Gorilla to executions of Sandglass. 
Then we intended to proceed by contradiction: assume that a correctness guarantee is violated in Gorilla, and map this violation to Sandglass; since correctness violations are not possible in Sandglass~\cite{Sandglass}, we could then conclude that neither they can be in Gorilla.

The best laid plans often go awry, and, as we discuss below, ours was no exception---but we were able to nonetheless retain the conceptual simplicity of our initial approach. 

		\subsection{The Main Story, and How it Fails}\label{sec:wall}
\sloppy
The mapping from Gorilla to Sandglass must satisfy certain {\em well-formedness} and {\em equivalence} conditions.
The former specify how to map a Gorilla execution into one that satisfies the Sandglass model~(SM) and follows the Sandglass protocol;
the latter allow us to map violations from Gorilla to Sandglass, {\em i.e.}, they preserve certain properties of the behavior of \textit{correct} nodes in Gorilla and reinterpret them as the behavior of \textit{good} nodes in Sandglass.

{\em Well-formedness} requires mapping correct nodes to good nodes, and Byzantine nodes to defective nodes, while respecting model constraints ({\em e.g.}, at each step defective nodes should be fewer than good nodes).
The first half of this mapping is easy: except for calculating a VDF, correct nodes in GM are not doing anything different than good nodes in SM.
Thus, mapping a step in~GM to a step in~SM yields a straightforward connection between correct and good nodes.
The second half, however, is trickier.
Defective nodes in~SM can suffer from benign faults like omission and crashing, but these fall short of fully capturing Byzantine behavior in~GM.
In particular, Byzantine nodes, even when sending valid messages, can violate the timing constraints that Gorilla places on a node's actions, \textit{e.g.}, by splitting the calculation of a single VDF into multiple steps.
Thus, before a Gorilla execution can be mapped to a Sandglass execution, Byzantine nodes' actions must be brought to conform to step boundaries and not spill across steps.
After tidying things up this way, it must become possible to map the faulty actions of the Byzantine nodes to a combination of crashes, omissions, and network delays, \emph{i.e.}, to the faults and anomalies that~SM allows.

\textit{Equivalence} in turn requires that, when mapping executions from Gorilla to Sandglass,  a correct node and its corresponding good node send and receive in every step messages that allow them to update their proposed value, round number, priority, and unanimity counter in the same way. Since messages play the same role in both protocols, this is sufficient for good nodes in Sandglass to decide identically to the corresponding correct nodes in Gorilla.

Our plan to realize this logical mapping involved splitting it into two concrete, intermediate mappings: a first mapping from an initial Gorilla execution to an intermediate Gorilla execution in which Byzantine actions conform to step boundaries; and a second mapping from that intermediate execution to a Sandglass execution.
We require all of our well-formedness and equivalence conditions to hold throughout these mappings: ($i$) model constraints must be always respected, ($ii$) correct nodes in the intermediate execution send and receive the equivalent (indeed, the same!)  messages as their counterparts in the initial execution, at the same steps, and ($iii$) good nodes in the final execution send and receive equivalent messages as their correct counterparts in the intermediate execution, at the same steps.

\begin{figure}
\centering
\begin{subfigure}{0.5\textwidth}
  \centering
  \captionsetup{justification=centering}
  \includegraphics[width=1\linewidth]{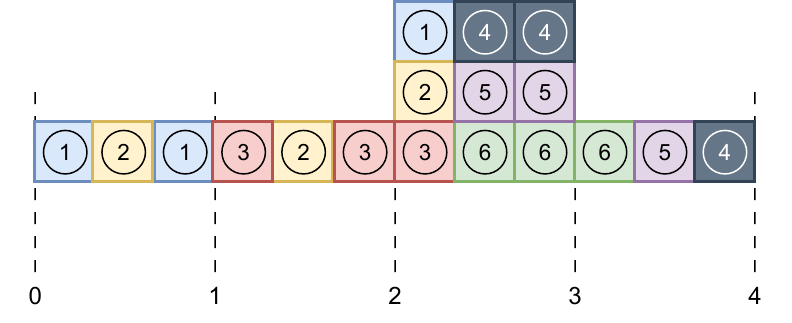}
  \caption{The counterexample.}
  \label{fig:wall}
\end{subfigure}%
\begin{subfigure}{0.5\textwidth}
  \centering
  \captionsetup{justification=centering}
  \includegraphics[width=1\linewidth]{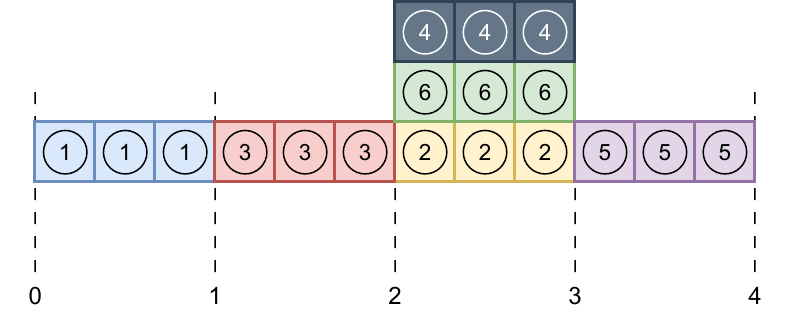}
  \caption{The solution enabled by peeking.}
  \label{fig:peek}
\end{subfigure}
\caption{An execution that cannot be reorganized in~GM (a), and how peeking solves the problem in~GM+ (b).}
\label{fig:fig}
\end{figure}

Unfortunately, {\em well-formedness and equivalence cannot be satisfied by the first mapping}.
To see why, consider Figure~\ref{fig:wall}.
Here, each square represents a VDF unit calculated by a Byzantine node for a specific input, denoted by a unique color.
Numbered circles represent the corresponding messages, {\em e.g.}, the VDF units containing \msg{1} are associated with message \msg{1}.
Each VDF calculation takes three ticks, and a step comprises three ticks.
The numbered dashed lines indicate the steps, {\em i.e.}, the three ticks between lines~$i$ and~$i+1$ belong to step~$i$.
Assume that, to maintain a majority of correct nodes in the system, the maximum allowable number of Byzantine nodes in the four steps shown in the figure are, respectively, 1, 1, 3, and 1.
Moreover, assume that messages \msg{4}, \msg{5}, and \msg{6} all include in their coffers messages \msg{1}, \msg{2}, and \msg{3}.
Finally, assume that messages \msg{4}, \msg{5}, and \msg{6} are sent to correct nodes at the start of Step 4. 
Since the actions of Byzantine nodes in Figure~\ref{fig:wall} do not conform to step boundaries, the first mapping should be able to organize them in a way that ensures that ($i$) 
correct nodes receive messages \msg{4}, \msg{5}, and \msg{6} at the beginning of Step 4, and ($ii$) each of these messages in turn includes messages \msg{1}, \msg{2}, and \msg{3}. Thus, the calculation of the VDFs for messages \msg{1}, \msg{2}, and \msg{3}  must be completed before those for  \msg{4}, \msg{5}, and \msg{6} can start.
Now, since steps 0 and 1 include only one Byzantine node, they can only accommodate one VDF, \emph{i.e.}, only one VDF can be calculated in each of steps 0 and 1. Without loss of generality, let those VDFs be \msg{1} and \msg{3}, respectively. VDF \msg{2} must still complete before messages \msg{4}, \msg{5}, and \msg{6}: thus, it has to be placed in Step 2. Note that, although Step 2 could accommodate two more Byzantine VDFs at Step 2, they cannot be placed there, since the completion of VDF \msg{2} must precede the start of the calculation of VDFs \msg{4}, \msg{5}, and \msg{6}: the earliest step where they can start is Step 3. However, it is impossible to accommodate all three there, since in Step 3 there is a single Byzantine node.

Our first attempt at mapping executions from Gorilla to Sandglass has thus failed.  Fortunately, though, it is possible to retain the strategy that underlies it and overcome the above counterexample without weakening our well-formedness and equivalence conditions. Instead, we proceed to weaken 
the model in which we operate, by giving Byzantine nodes extra power.

\subsection{A New Beginning}

\newcommand{\textForProofbeforeFirstDef}{
The first step in our two-step process for mapping a Gorilla execution $\eta_G$ into a Sanglass execution $\eta_S$ is to reorganize the actions taken by Byzantine nodes in~$\eta_G$: we want to map~$\eta_G$ to an execution where Byzantine nodes join the system and receive valid messages at the beginning of a step (by the first tick) and broadcast valid messages and leave the system at the step's end (at its $K$-th tick).
Since, as explained in Section~\ref{sec:wall}, satisfying all of these requirements is not possible, we extend~GM to a new model.

% The first step in our two-step process for mapping a Gorilla execution $\eta_G$ into a Sanglass execution $\eta_S$ is to reorganize the actions taken by Byzantine nodes in $\eta_G$: we want to map~$\eta_G$ to an execution where Byzantine nodes join the system and receive valid messages at the beginning of a step (by the first tick) and leave the system and broadcast valid messages at the step's end (at its $K$-th tick).
% Since, as explained in Section~\ref{sec:wall}, satisfying all of these requirements is not possible, we extend~GM to a new model.

We need some way to calculate a VDF on an input that includes the final result of VDF calculations that are still in progress. To achieve this, we extend the oracle's API to allow Byzantine nodes to  \emph{peek} at those future outcomes. By issuing the oracle's {\em peek} query, Byzantine nodes active in any step~$s$ can learn the result of a VDF computed by Byzantine nodes finishing in step~$s$ even before its calculation has ended.

We thus introduce GM+, a model that extends GM by having a new oracle,~$\Omega^+$, that supports one additional method:

\begin{description}
	\item[Peek($\gamma$):]  immediately returns~$\textit{vdf}_\gamma$.
\end{description}
%Peek($\gamma$) allows to shortcut, subject to certain restrictions outlined below, the set of $K$ Get() calls required to compute {\em vdf$_\gamma$} in the GM model. 

In any tick, a 
Byzantine node in~\GMPLUS
can call Peek() multiple times, with different inputs. 
However, Byzantine nodes can only call Peek subject to two
conditions: 
\begin{itemize}
	
%\item A node can peek $\textit{vdf}_\gamma$ only if, eventually,  some subset of nodes performs the $K$ calls to Get() necessary to produce the $K$ units of $\textit{vdf}_\gamma$.

\item A Byzantine node can peek in step~$s$ at $\textit{vdf}_\gamma$ only if Byzantine nodes commit to finish the VDF calculation for input~$\gamma$ within~$s$; and
\item a Byzantine node does not peek at $\textit{vdf}_\gamma$, where $\gamma = (M,nonce)$, if~$M$ in turn contains some VDF result~$v$ obtained by peeking, and the calculation of~$v$ has yet to finish in this tick. 
\end{itemize}
Note that these restrictions only limit the {\em additional} powers that \GMPLUS grants the adversary: in GM+, Byzantine nodes remain strictly stronger than in GM.

With this new model, we first map an execution of Gorilla in~GM to an execution of Gorilla in GM+, in which Byzantine behavior is reorganized with the addition of peeking.
Hence follows the first lemma of our scaffolding: the existence of the first mapping.
} % this "}" is for 
\textForProofbeforeFirstDef

\begin{restatable}{definition}{defeqGGP}\label{def:eq-G-GP}
	 Consider an execution $\eta_G$ in~GM and an execution $\eta_G^+$ in~GM+.
     We say~$\eta_G^+$ is a reorg of~$\eta_G$ iff the following conditions are satisfied:

    \begin{description}
        \item[\textsc{Reorg}-\mbox{$1$}] For every correct node $p$ in $\eta_G$, there exists a correct node $p^+$ in $\eta_G^+$,  such that $p$ and $p^+$ ($i$)~join and leave the system at the same ticks in the same steps and ($ii$)~receive and send the same messages at the same ticks in the same steps.
        			
        \item[\textsc{Reorg}--\mbox{$2$}] Each Byzantine node in~$\eta_G^+$ ($i$)~joins at the first tick of a step and leaves after the last tick of that step; ($ii$) receives messages at the first tick of a step and sends messages at the last tick of that step; and (iii)~sends and receives only valid messages.
        
        \item[\textsc{Reorg}-\mbox{$3$}] If in~$\eta_G$ a Byzantine node sends a valid message~$m$ at a tick in step~$s$, then in~$\eta_G^+$ a Byzantine node sends~$m$ at a tick in some step~$s' \le s$. 	
    \end{description}

\end{restatable}

\begin{restatable}{lemma}{lemmareorg} \label{lemma:reorg}
There exists a mapping \textsc{Reorg} that maps an execution~$\eta_G$ in~GM to an execution~$\eta_G^+$ in~\GMPLUS, denoted~$\eta_G^+ = \textsc{Reorg}(\eta_G)$, such that~$\eta_G^+$ is a reorg of~$\eta_G$.
\end{restatable}

\newcommand{\textBetweenlemmareorgAnddefMAP}{
While peeking solves the challenge with reorganizing Byzantine behavior, it complicates our second mapping. The ability to peek granted to Byzantine nodes in GM+ has no equivalent in Sandglass -- it simply cannot be reduced to the effects of network delays or to the  behavior of defective nodes. Therefore, we weaken~SM so that defective nodes can benefit from a capability equivalent to peeking.

We do so by introducing SM+, a model that is identical to~SM, except for the following change: defective nodes at step~$s$ can receive any message~$m$ sent by a defective node no later than~$s$ -- as opposed to~$(s-1)$ in~SM~--~as long as~$m$ does not contain in its coffer a message that is sent at~$s$. Note that allowing defective nodes to receive in a given step a message $m$ sent by defective nodes within that very step maps to allowing Byzantine nodes to peek at a message whose~$vdf$ will be finished by Byzantine nodes within the same step; and the constraint that~$m$ shouldn't contain in its coffer other messages sent in the same step, maps to the constraint that Byzantine nodes cannot peek at messages whose coffer also contains a peek result from the same step.

One might rightfully ask: But the plan to leverage the correctness of Sandglass in~SM? Indeed, but fortunately, {\em Sandglass still guarantees deterministic agreement and termination with probability~1 under the SM+ model} (\S\ref{sec:appendix-SM-proof}). Thus, it is suitable to map a Gorilla execution in~GM+ to a Sandglass execution in~SM+, and orient our proof by contradiction with respect to the correctness of Sandglass in~SM+.

Formally, we specify our second mapping as follows. We map messages by simply translating the data structure:
} % This "}" is for 
\textBetweenlemmareorgAnddefMAP
\begin{restatable}{definition}{defMAP}\label{def:MAP}
Given a message $m$ in the Gorilla protocol, the mapping $\textsc{Mapm}$ produces a message in the Sandglass protocol as follows
\begin{enumerate}
\item Omit the~\textit{vdf} and the nonce from $m$.
\item Let~$p_i$ be the node that sends~$m$. Include $p_i$ as a field in $m$.
\item If $m$ is the $j$-th message sent by $p_i$, add a field $\textit{uid} = j$ to $m$.
\item Repeat the steps above for all of the messages in $m$'s coffer.
\end{enumerate}
Denote the result by $\hat{m} = \textsc{Mapm}(m)$. 
We say~$m$ and~$\hat{m}$ are \emph{equivalent}.
Furthermore, with a slight abuse of notation, we apply $\textsc{Mapm}$ to a set of messages as well, \emph{i.e.}, if $\mathcal{M}$ is a set of messages, and we map each message $m\in \mathcal{M}$, we obtain the message set $\textsc{Mapm}(\mathcal{M})$.
\end{restatable}

Thus, we can define the execution mapping:
\begin{restatable}{definition}{defeqLVS}\label{def:eq-L-VS}
	 Consider an execution $\eta_G^+$ in~GM+ and an execution $\eta_S^+$ in~SM+.
     We say~$\eta_S^+$ is an interpretation of~$\eta_G^+$ iff the following conditions are satisfied:
	\begin{enumerate}
		\item \label{def:equi:1} The nodes in $\eta_G^+$ are in a one-to-one correspondence with the nodes in $\eta_S^+$.
		For every node $p$ in $\eta_G^+$, we denote the corresponding node in $\eta_S^+$ with $\hat{p}$.
		\item \label{def:equi:2} Nodes~$p$ and $\hat{p}$ join and leave at the same steps in $\eta_G^+$ and $\eta_S^+$, respectively.
		Furthermore, their initial values are the same.
          \item \label{def:equi:3} If $p$ is a Byzantine node, then $\hat{p}$ is defective in~SM+; otherwise, $\hat{p}$ is a good node in~SM+. 
        \item \label{def:equi:4} Node~$\hat{p}$ sends~$\hat{m}$ at step~$s$ in~$\eta_S^+$ iff~$p$ generates a message~$m$ in~$\eta_G^+$ at step~$s$. 
		Note that in~$\eta_G^+$, correct nodes send their messages to all as soon as they are generated, while Byzantine nodes may only send their messages to a subset of nodes once their messages are generated.
		\item \label{def:equi:5} 
        Node~$\hat{p}$ receives~$\hat{m}$ at step~$s$ in~$\eta_S^+$ iff~$p$ receives~$m$ at step~$s$ in~$\eta_G^+$.		
	\end{enumerate}
\end{restatable}

\begin{restatable}{lemma}{lemmainterpret}\label{lemma:interpret}
    Consider any execution~$\eta_G$ in~GM, and an execution~$\eta_G^+$ in~GM+ is a reorg of~$\eta_G$.
    There exists a mapping \textsc{Interpret} that maps~$\eta_G^+$ to an execution~$\eta_S^+$ in~SM+, denoted as~$\eta_S^+ = \textsc{Interpret}(\eta_G^+)$, such that~$\eta_S^+$ is an interpretation of~$\eta_G^+$.
\end{restatable}
Finally, for our proof by contradiction to work, we have to show that Sandglass is correct in SM+.
The proof is deferred to \S\ref{sec:sm:proof}.
\begin{restatable}{theorem}{theoremSM}\label{theorem:SM+}
Sandglass satisfies agreement and validity deterministically and termination with probability~1 in SM+.
\end{restatable}

\subsubsection{Safety}\label{sec:correctness-safety}
\newcommand{\textAtStartOfSafety}{
We prove that Gorilla satisfies Validity and Agreement.
The proofs follow the same pattern: assume a violation exists in some execution $\eta_G$ of Gorilla running in GM; map that execution to $\eta_G^+ = \textsc{Reorg}(\eta_G)$ in GM+; then,
map $\eta_G^+$ again to $\eta_S^+ = \textsc{Interpret}(\eta_G^+)$ in SM+; and, finally, rely on the fact that these mappings ensure that correct nodes in $\eta_G$ and good nodes in $\eta_S^+$ reach the same decisions in the same steps to derive a contradiction.
} % this "}" is for 
\textAtStartOfSafety
This approach is made rigorous in following lemmas, proved in \S\ref{sec:appendix-safety}.
\begin{restatable}{lemma}{lemmadecisionreorg}\label{lemma:decision-reorg}
    Consider an arbitrary Gorilla execution~$\eta_G$, and~$\eta_G^+ = \textsc{Reorg}(\eta_G)$.
    If a correct node $p$ decides a value $v$ at step $s$ in $\eta_G$, then~$p$'s corresponding node~$p^+$ decides~$v$ at step $s$ in $\eta_G^+$.
\end{restatable}

\begin{restatable}{lemma}{lemmadecisioninterpret}\label{lemma:decision-interpret}
Consider any execution~$\eta_G$ in~GM. 
If an execution~$\eta_S^+$ in SM+ is an interpretation of an execution $\eta_G^+ = \textsc{Reorg}(\eta_G)$ in GM+, then the following statements hold:
\begin{enumerate}

	\item \label{lemma:equivalence:statement1} 
If a correct node $p$ decides a value $v$ at step $s$ in $\eta_G^+$, then the corresponding~$\hat{p}$, decides~$v$ at step $s$ in $\eta_S^+$.

\item \label{lemma:equivalence:statement2} 
Consider the first message~$m = (r,v,priority, \varpriorityCounter, M, nonce, \textit{vdf})$ that~$p$ generates for round~$r$. 
Let the step when~$m$ is generated be~$s$. 
If~$\varpriorityCounter$ is~0, then~$\hat{p}$ randomly chooses value~$v$ as the proposal value at step~$s$ in~$\eta_S^+$.
\end{enumerate}
\end{restatable}
$ $\\
We can now state and prove the safety guarantees.

\begin{restatable}{theorem}{thmgorillaAgr}
Gorilla satisfies agreement in~GM.
\end{restatable}

\begin{restatable}{proof}{proofthmgorillaAgr}
By contradiction, assume that there exists a Gorilla execution~$\eta_G$ in~GM that violates agreement.
This means that there exist two correct nodes $p_1$ and $p_2$, two steps~$s_1$ and~$s_2$, and two values~$v_1\neq v_2$ such that~$p_1$ decides~$v_1$ at~$s_1$ and~$p_2$ decides~$v_2$ at~$s_2$.
Consider~$\eta_G^+ = \textsc{Reorg}(\eta_G)$.
According to Lemma~\ref{lemma:decision-reorg},~$p_1^+$ decides~$v_1$ at~$s_1$ and~$p_2^+$ decides~$v_2$ at~$s_2$, in~$\eta_G^+$.
Now, consider~$\eta_S^+ = \textsc{Interpret}(\eta_G^+)$.
According to Lemma~\ref{lemma:decision-interpret},~$\hat{p}_1^+$ decides~$v_1$ at~$s_1$ and~$\hat{p}_2^+$ decides~$v_2$ at~$s_2$, in~$\eta_S^+$.
However, this contradicts the fact Sandglass satisfies agreement in SM+ (Theorem~\ref{theorem:SM+}).
Therefore, Gorilla satisfies agreement in~GM.
\end{restatable}

\begin{restatable}{theorem}{thmgorillaVad}
Gorilla satisfies validity in~GM.
\end{restatable}

\begin{restatable}{proof}{proofthmgorillaVad}
By contradiction, assume that there exists a Gorilla execution~$\eta_G$, such that ($i$) all nodes that ever join the system have initial value~$v$; ($ii$) there are no Byzantine nodes; and ($iii$) a correct node~$p$ decides~$v' \neq v$.

Since~GM+ is an extension of~GM, $\eta_G$ conforms to~GM+.
According to Definition~\ref{def:eq-G-GP},~$\eta_G^+ = \eta_G$ in~GM+ is trivially a reorg of~$\eta_G$.
Consider~$\eta_S^+ = \textsc{Interpret}(\eta_G^+)$.

By the construction of the \textsc{Interpret} mapping (in Lemma~\ref{lemma:interpret}), good nodes in~$\eta_S^+$ have the same initial values as their corresponding correct nodes in~$\eta_G$.
Furthermore, since there are no Byzantine nodes in~$\eta_G^+$, there are no defective nodes in~$\eta_S^+$ by Definition~\ref{def:eq-L-VS}.
Therefore, by Validity of Sandglass in SM+ (Theorem~\ref{theorem:SM+}), no good node decides~$v' \neq v$.
However, by Lemma~\ref{lemma:decision-reorg} and Lemma~\ref{lemma:decision-interpret},~$p$ decides~$v' \neq v$, which leads to a contradiction.
Therefore, Gorilla satisfies validity in~GM.
\end{restatable}

%Having proven Gorilla's safety properties, we now address Liveness. 

\subsection{Liveness}
\newcommand{\textAtStartOfLiveness}{
Similar to the safety proof, the liveness proof proceeds by contradiction: it starts with a liveness violation in Gorilla, and maps it to a liveness violation in Sandglass.

Formalizing the notion of violating termination with probability~1 requires specifying the probability distribution used to characterize the probability of termination.
To do so, we first have to fix all sources of non-determinism~\cite{aspnes-randomized, foundations-probabilistic-programming,kaminski}.
For our purposes, non-determinism in~GM and~GM+ stems from correct nodes, Byzantine nodes and their behavior; in~SM+, it stems from good nodes, defective nodes and the scheduler.

For correct, good, and defective nodes, non-determinism arises from the joining/leaving schedule and the initial value of each joining node.
For Byzantine nodes in GM and GM+, fixing non-determinism means fixing their action \emph{strategy} according to the current history of an execution.
Similarly, fixing the scheduler's non-determinism means specifying the timing of message deliveries and the occurrence of benign failures, based on the current history.
We, therefore, define non-determinism formally in terms of  an environment and a strategy.

To this end, we introduce the notion of a \emph{message history}, and define what it means for a set of messages exchanged in a given step to be \emph{compatible} with the  message history that precedes them.
} % this "}" is for 
\textAtStartOfLiveness

\begin{restatable}{definition}{defMesHis}
    For any given execution in~GM and~GM+ (resp.,~SM+), and any step~$s$, the \emph{message history up to~$s$},~$\mathcal{MH}_s$, is the set of~$\langle m, p, s'\rangle$ triples such that~$p$ is a correct node (resp., good node) and~$p$ receives~$m$ at~$s'\leq s$.
\end{restatable}
\begin{restatable}{definition}{defCon}
    We say a set~$\mathcal{MP}_{s+1}$ of~$\langle m, p, s+1\rangle$ triples is \emph{compatible} with a message history up to~$s$,~$\mathcal{MH}_s$, if there exists an execution such that for any~$\langle m, p, s+1\rangle \in \mathcal{MP}_{s+1}$, the correct node (resp., good node)~$p$ receives~$m$ at step~$(s+1)$.
\end{restatable}

\begin{restatable}{definition}{defEnvG}
An environment~$\mathcal{E}$ in~GM and~GM+ (resp.,~SM+) is a fixed joining/leaving schedule and fixed initial value schedule for correct nodes (resp., good and defective nodes).
\end{restatable}

\begin{restatable}{definition}{defStrategyB}
	Given an environment $\mathcal{E}$, a strategy $\Theta_{\mathcal{E}}$ for the Byzantine nodes (resp., scheduler) in~GM and~GM+ (resp.,~SM+) is a function that takes the message history~$\mathcal{MH}_s$ up to a given step~$s$ as the input, and outputs a set $\mathcal{MP}_{s+1}$ that is compatible with~$\mathcal{MH}_s$.
\end{restatable}

\newcommand{\textRandomizedStrategyAli}{
Before proceeding, there is one additional point to address.
The most general way of eliminating non-determinism is to introduce randomness through a fixed probability distribution over the available options. However, the following lemma, proved in  \S\ref{sec:appendix-liveness}, establishes that Byzantine nodes do not benefit from employing such a randomized strategy.
} % this "}" is for 
\textRandomizedStrategyAli
\begin{restatable}{lemma}{lemmastrategyRD}
For any environment~$\mathcal{E}$, if there exists a \emph{randomized} Byzantine strategy for Gorilla that achieves
a positive non-termination probability, then there exists a \emph{deterministic} Byzantine
strategy for Gorilla that achieves a positive non-termination probability.
\end{restatable}

\newcommand{\textAfterDefEnv}{

Since the output~$vdf$ of a call to the VDF oracle is a random number, the~$(\textit{vdf}\mod 2)$ operation in line~\ref{gorilla:vdfValueUpdate} of Gorilla is equivalent to tossing an unbiased coin.
Given a strategy~$\Theta_{\mathcal{E}}$,\footnote{When it is clear from the context, we will omit the environment from the subscript of the strategy.} the nodes might observe different coin tosses as the execution proceeds; thus, the strategy specifies the action of the Byzantine nodes for all possible coin toss outcomes. The scheduler's strategy in SM+ is similarly specified for all coin toss outcomes.
Therefore, once a strategy is determined, it admits a \emph{set} of different executions based on the coin toss outcomes; we denote it by~$H_{\Theta}$.
Specifically, a strategy determines an action for each outcome of any coin toss.

Given a strategy~$\Theta$, we can define a probability distribution~$P_{H_{\Theta}}$ over~$H_\Theta$.
For each execution~$\eta\in H_\Theta$, there exists a unique string of zeros and ones, representing the coin tosses observed during~$\eta$.
Denote this bijective correspondence by~$\textsc{Coins}:H_\Theta\rightarrow\{0,1\}^*\cup\{0,1\}^\infty$, and the probability distribution on the coin toss strings in~$\textsc{Coins}(H_\Theta)$ by~$\tilde{P}_{H_{\Theta}}$. 
For every event~$E\subset H_\Theta$, if ~$\textsc{Coins}(E)$ is measurable in~$\textsc{Coins}(H_\Theta)$, then~$\tilde{P}_{H_{\Theta}}(\textsc{Coins}(E))$ is well-defined; thus,~$P_{H_{\Theta}}(E)$ is also well-defined and $P_{H_{\Theta}}(E)=\tilde{P}_{H_{\Theta}}(\textsc{Coins}(E))$.
We denote~$P_{H_{\Theta}}$ as the probability distribution induced over~$H_\Theta$ by its coin tosses.%xxx

Equipped with these definitions, we can formally define termination with probability~1.}
\textAfterDefEnv % this "}" is for 
\begin{restatable}{definition}{defTermProbOne}
    The Gorilla protocol terminates with probability~1 iff for every environment~$\mathcal{E}$ and every Byzantine strategy~$\Theta$ based on~$\mathcal{E}$, the
    probability of the termination event~$T$ in~$H_{\Theta}$, \emph{i.e.},~$P_{H_{\Theta}}(T)$, is equal to 1.
\end{restatable}
\newcommand{\textAfterdefTermProbOne}{
This definition gives us the recipe for proving by contradiction that Gorilla terminates with probability 1. We first assume there exists a Byzantine strategy~$\Theta$ that achieves a non-zero non-termination probability, and map this strategy through the~$\textsc{Reorg}$ and~$\textsc{Interpret}$ mappings to a scheduler strategy~$\Lambda$ that achieves a non-zero non-termination probability in~SM+.
However,~$\Lambda$ cannot exist, as the Sandglass protocol terminates with probability~1 in~SM+~(Theorem~\ref{theorem:SM+}).
} % this "}" is for 
\textAfterdefTermProbOne

\begin{restatable}{lemma}{lemmareorgstrategy}\label{lemma:reorg-strategy}
	If there exists an environment~$\mathcal{E}$ and a Byzantine strategy~$\Theta_{\mathcal{E}}$ in~GM that achieves a positive non-termination probability,
	then there exists an environment~$\mathcal{E}'$ and a Byzantine strategy~$\Psi_{\mathcal{E}'}$ in~GM+ that also achieves a positive non-termination probability.
\end{restatable}
\begin{restatable}{proof}{prooflemmareorgstrategy}
	Assume there exist an environment~$\mathcal{E}$ and a Byzantine strategy~$\Theta_{\mathcal{E}}$ in~GM that achieves a positive non-termination probability.
        Consider the~$\textsc{Reorg}$ mapping.
        Since, according to Lemma~\ref{lemma:reorg}, the joining/leaving and initial value schedules for correct nodes remain untouched by the~$\textsc{Reorg}$ mapping, we just set~$\mathcal{E}' = \mathcal{E}$.
        In the rest of the proof, we omit the environments for brevity.

        We now show that the strategy~$\Psi$ exists, and is in fact the same as~$\Theta$.
        For brevity, let~$R_\Theta$ denote~$\textsc{Reorg}(H_\Theta)$, and consider any execution~$\eta$ in~$H_\Theta$.
        By Lemma~\ref{lemma:reorg}, correct nodes in~$\eta$ receive the same messages, at the same steps, as the correct nodes in~$\textsc{Reorg}(\eta)$ and, moreover, the coin results in~$\eta$ are exactly the same as the ones in~$\textsc{Reorg}(\eta)$.
        Thus, the message history of correct nodes up to any step~$s$ in~$\eta$ is the same as the message history of correct nodes up to the same step in~$\textsc{Reorg}(\eta)$.
        In addition, because~$\textsc{Reorg}(\eta)$ is a~GM+ execution, compatibility is trivially satisfied.
        Thus, we conclude that Byzantine nodes in~$R_\Theta$ follow the same strategy as in~$\Theta$, conforming to the same coin toss process.
        Let us denote this strategy with~$\Psi$.

        Note that according to Lemma~\ref{lemma:decision-reorg}, whenever a correct node decides at some step~$s$ in~$\eta$, its corresponding correct node in~$\textsc{Reorg}(\eta)$ decides the same value at the same step.
        Therefore, the set of non-terminating executions in~$H_\Theta$ are mapped to the set of non-terminating executions in~$R_\Theta$ in a bijective manner.
        Let us denote these sets as~$NT_H$ and~$NT_R$, respectively.
        Since the same coin toss process induces probability distributions~$P_{H_\Theta}$ and~$P_{R_\Theta}$ on~$H_\Theta$ and~$R_\Theta$, respectively, we conclude that~$P_{H_\Theta}(NT_H) = P_{R_\Theta}(NT_R)$.
        Therefore, since~$P_{H_\Theta}(NT_H) > 0$ by assumption, this concludes our proof, as we have shown the existence of a strategy~$\Psi$ in~GM+ that achieves a positive non-termination probability.
        \qedhere 
\end{restatable}

\smallskip
\noindent % Just do it doesn't seem like part of the description bullets above
A similar lemma applies to the second mapping.
We prove it in \S\ref{sec:appendix-liveness}.
\begin{restatable}{lemma}{lemmainterpretstrategy}\label{lemma:interpret-strategy}
	If there exists an environment~$\mathcal{E}$ and a strategy $\Psi$ for Byzantine nodes in~GM+ that achieves a positive non-termination probability,
	then there exists an environment~$\mathcal{E}'$ and a scheduler strategey~$\Lambda_{\mathcal{E}'}$ in~SM+ that also achieves a positive non-termination probability.
\end{restatable}

Based on these lemmas, we are finally ready to prove Gorilla's liveness guarantee.

\begin{restatable}{theorem}{thmtermprobone}
    The Gorilla protocol terminates with probability~1.
\end{restatable}
\begin{restatable}{proof}{proofthmtermprobone}
By contradiction, assume that there exist a GM environment and a Byzantine strategy~$\Theta$ in Gorilla that achieve a positive non-termination probability.
    By Lemma~\ref{lemma:reorg-strategy}, there exist 
    a GM+ environment and
    a strategy~$\Psi$ for the Byzantine nodes in~GM+ that achieve a positive non-termination probability.
    Similarly, by Lemma~\ref{lemma:interpret-strategy}, there exists 
    an SM+ environment and 
    a scheduler strategy~$\Lambda$ in~SM+ that achieve a positive non-termination probability.
    But this is a contradiction, since Sandglass 
    terminates with probability~1 in~SM+ (Theorem~\ref{theorem:SM+}).
    Thus, Byzantine strategy $\Theta$ cannot force a positive non-termination probability; Gorilla 
    terminates with probability~1.
\end{restatable}

%% file: conclusion.tex
    \section{Conclusion} \label{sec:conclusion}

Gorilla Sandglass is the first Byzantine-tolerant consensus protocol to guarantee, in the same synchronous model adopted by Nakamoto, deterministic agreement and termination with probability~1 in a permissionless setting.
To this end, Gorilla leverages VDFs to extend the approach of Sandglass, the first protocol to provide similar safety guarantees in the presence of benign failures. Neither Gorilla nor Sandglass are practical protocols, however: they exchange a very large number of messages and the number of rounds they require to decide is large even under favorable circumstances, and can, in general, be exponential. Is there a {\em practical} permissionless protocol that can achieve deterministic safety and tolerate fewer than a half Byzantine nodes?

%% file: sandglassplus.tex
\section{Sandglass Plus}
\label{sec:sm}

In this section, we first introduce the~SM+ model, which is nearly identical to that of Sandglass~\cite{cryptoeprint:2022/796}, with the exception that it permits defective nodes to receive messages sent by other defective nodes within the same step, subject to certain constraints.

We then prove that Sandglass remains correct in this new model by satisfying Validity, Agreement, and termination with probability 1.
Fortunately, the correctness proof of Sandglass requires only minor modifications to the proofs of two lemmas to be applicable to the~SM+ model.
We will show in Section~\ref{sec:sm:proof} the two lemmas that require change, and refer the readers to the Sandglass paper for the rest of the proof.

\subsection{The SM+ Model}
    SM+ and the Sandglass model (SM) largely make the same set of assumptions. 
    We show the only difference here and refer the readers to the rest of the model in the Sandglass paper~\cite{cryptoeprint:2022/796}:
    SM assumes that \textit{if in step~$t$ a node~$p_i$ receives message~$m$ with~$Receive_i$, then~$m$ was sent in some step~$t' < t$}.
    In~SM+, we weaken this assumption, by allowing defective nodes to non-recursively receive messages sent from defective nodes within the same step. 
    Formally, if in step~$t$ a node~$p_i$ receives message~$m$ from~$p_j$ with~$Receive_i$, then~$m$ was sent in some step~$t' < t$ when at least one of~$p_i$ and~$p_j$ is good, or~$t' \le t$ when both~$p_i$ and~$p_j$ are defective and~$m$ does not contain in its coffer a message that is also sent in~$s$.

\subsection{Sandglass protocol}
We run Sandglass protocol (Algorithm~\ref{protocol:sandglass}) in~SM+.

 \begin{algorithm}[h]
  	\floatname{algorithm}{Protocol}
  	\caption{Sandglass: Code for node $p_i$}
  	
  	\label{protocol:sandglass}
  	
  	\begin{algorithmic}[1]
  		\Procedure{Init}{$input_i$} \label{sandglass:Init}
  		\State{~$v_i\gets input_i$;~$priority_i\gets 0$;~$\varpriorityCounter_i \gets 0$;~$r_i=1$;~$M_i = \emptyset$;~$Rec_i = \emptyset$;~$\emph{uid}_i = 0$}
  		\label{sandglass:initialsetup}
  		\EndProcedure
  		
  		\Procedure{step}{} \label{sandglass:step}
  		
  		\ForAll{$m=(\cdot,\cdot,\cdot,\cdot,\cdot,M)$ received by~$p_i$ }\label{sandglass:forunionReceivedMessages}
  		\State $Rec_i \gets Rec_i \cup \{m\} \cup M $\label{sandglass:unionReceivedMessages}
  		\EndFor
  		
  		\If{$\max_{|Rec_i(r)| \ge \varR}(r) \ge r_i$ \label{sandglass:roundNumberMax}}

  		\State $r_i =  \max_{|Rec_i(r)| \ge \varR}(r) + 1$ \label{sandglass:enterNewRound}

  		\State $M_i = \emptyset$ \label{sandglass:resetM}
  		\ForAll{$m=(\cdot, r_i-1,\cdot,\cdot,\cdot,M) \in Rec_i(r_i-1)$ } \label{sandglass:unionCollectedMessagesStart}
  		\State $M_i \gets M_i \cup \{m\}\cup M$\label{sandglass:unionCollectedMessagesEnd}
  		\EndFor		
  		
  		\State Let~$C$ be the multi-set of messages in $M_i(r_i-1)$ with the largest priority. 
  		\label{sandglass:multiset}
  		
  		\If{all messages in $C$ have the same value~$v$} \label{sandglass:valueUpdateStart}
  		\State $ v_i \gets v$
  		\Else
  		\State $v_i \gets \text{one of} \{a,b\},\text{chosen uniformly at random}$\label{sandglass:valueUpdateEnd}
  		\EndIf
  		
  		\If{all messages in $M_i(r_i-1)$ have the same value~$v_i$}
  		\label{sandglass:checkValues}
  		\State{$\varpriorityCounter_i \gets 1 + \min\{\varpriorityCounter | (\cdot,r_i-1,v_i,\cdot,\varpriorityCounter, \cdot) \in M_i(r_i-1)\}$} 
  		\label{sandglass:setPriorityCounter}
  		\Else
  		\State{$\varpriorityCounter_i \gets 0$} \label{sandglass:resetPC}
  		\EndIf
  		
  		\State $priority_i \gets \max (0,\left\lfloor \frac{\varpriorityCounter_i}{\varR} \right\rfloor -5)$ \label{sandglass:setPriority}
  		
  		\If{$\varpriority_i \ge 6\varR+4$} \label{sandglass:decideStart}
  		\State{Decide$_i$($v_i$)} \label{sandglass:decidePoint}	
  		\EndIf	\label{sandglass:decideEnd}

  		\EndIf

  		\State $\emph{uid}_i \gets \emph{uid}_i+1;$ \label{sandglass:unique:Message}
  		\State $M_i \gets M_i \cup Rec_i(r_i)$
  		\label{sandglass:unionSameRoundMsg}
  		\State broadcast~$(p_i,\emph{uid}_i,r_i,v_i,priority_i, \varpriorityCounter_i, M_i)$ \label{sandglass:broadcast}
  		\EndProcedure
  		
  	\end{algorithmic}
  	
  \end{algorithm}

\subsection{Sandglass is Correct in SM+}\label{sec:appendix-SM-proof}
\label{sec:sm:proof}
The proof of correctness for Sandglass in~SM+ closely resembles the proof for~SM, with the exception of Lemma~5 and Lemma~11 in the Sandglass paper~\cite{cryptoeprint:2022/796}. 
While the statements of these lemmas remain unchanged, their proofs require minor modifications.

It is perhaps surprising that so little of the proof needs changing when moving from SM to SM+. The reason is that the original Sandglass proof assumes that all the messages generated by defective nodes contribute to their progress, regardless of when they are actually received. Specifically, when estimating the maximum possible round that defective nodes can be in, the proof  considers the totality of messages generated by defective nodes during the  execution and divides it by the size of the threshold of messages that must be received to advance to a new round. This best-case scenario for defective nodes already accounts for the additional flexibility that SM+ affords to defective nodes.  
In particular, the original proof already accounts for the possibility, allowed in SM+, that defective nodes receive, in a given step, messages that defective nodes  sent in that same step--even though SM disallows such executions. 

In Figure~\ref{fig:statementGraph}, which illustrates the dependency between the statements in the Sandglass paper~\cite{cryptoeprint:2022/796}, we highlight Lemma~5 and Lemma~11.
The letters~L, O, and C, signify Lemma, Observation, and Corollary, respectively.

\begin{figure}[h] 
  \caption{The structure of the proof through the dependencies among
    its constituent lemmas, corollaries, and observations.
    Preparatory results are in
    black; red and blue denote facts used in the proofs of Agreement and
    Termination, respectively.
    (Taken from Sandglass~\cite{cryptoeprint:2022/796}.)}
  \label{fig:statementGraph}
	\centering
        \includegraphics[width=\textwidth]{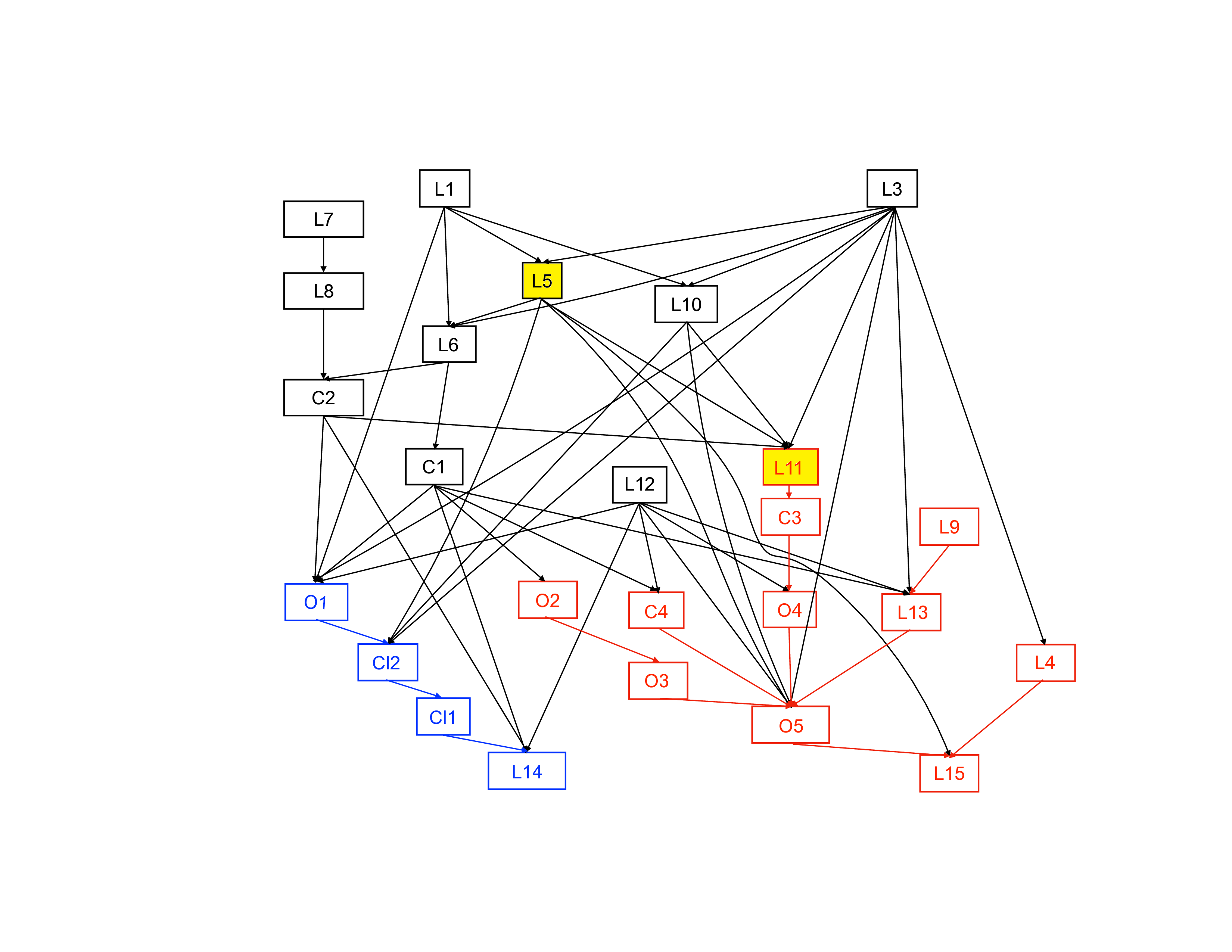}
\end{figure}

\label{sec:lemmas}
\newtheorem*{lemma5}{Lemma 5 in Sandglass~\cite{cryptoeprint:2022/796}}
\begin{lemma5} \label{lemma:untimelyOneRoundAhead}
	At any step~$T$, any defective node is at most one round ahead of any good node.
\end{lemma5}

\begin{proof}
	By contradiction. Assume that there exists an earliest
        step,~$T$, where some defective node~$p$ is more than one round
        ahead of a good node~$p_g$, {\em i.e.,} at~$T$ node~$p$ is in some round~$r$ and node~$p_g$ is in round~$r_{p_g} \le (r-2)$. 
	
	Note that no good node is in round~$(r-1)$ or larger before~$T$; otherwise, by Lemma~3 in Sandglass~\cite{cryptoeprint:2022/796}
    %\ref{lemma:goodNodeOneDeltaApart}
    , all good nodes would be in round~$(r-1)$ or larger at~$T$, contradicting~$r_{p_g} \le (r-2)$.
	Therefore, defective node~$p$ received no messages from  good nodes for round~$(r-1)$ by~$T$.
	%        , i.e. all the messages~$p$ received for Round~$(r-1)$ are from defective nodes, and the amount of which is at least~$\mathcal{T}$. 
	
	Consider the earliest step~$T' \le T$ where some defective node
        is in round~$(r-1)$.  Since~$T$ is the first step where some
        defective node is more than a round ahead of a good node, all good nodes
        must be in round~$(r-2)$ or larger at~$T'$; but, as we just
        showed, no good node is in round~$(r-1)$ or larger
        before~$T$. Therefore, all good nodes must be in round~$(r-2)$
        from~$T'$ until~$T$.
	
	Consider the~$k$ consecutive steps from~$T'$ to~$T$.  Let
        the number of messages generated by good nodes and defective nodes
        in each step be, respectively,~$g_1,...,g_{k}$
        and~$d_1, ..., d_{k}$.
        Since up to and including step~$T$ node~$p$ has
        received for round~$(r-1)$ only messages from defective nodes, and
        yet~$p$ is in round~$r$ at~$T$, by
        line~\ref{sandglass:roundNumberMax} of
        Sandglass,~$\Sigma_{i=1}^{i=k}d_i-1\ge {\mathcal T}$ and thus, by
        Lemma~1 in Sandglass~\cite{cryptoeprint:2022/796}
        %\ref{lemma:selectR}
        ,~$\Sigma_{i=2}^{i=k-1}g_i \ge {\mathcal T}$.  
        Since by
        assumption every step includes at least one good node ({\em
          i.e.,}~$g_1 > 0$), we have that~$\Sigma_{i=1}^{i=k-1}g_i > {\mathcal T}$.  Recall that
          during these~$(k-1)$ steps all good nodes are in round~$(r-2)$;
          then, all messages~$g_1,...,g_{k-1}$ are for round~$(r-2)$ and
          will all be received by all good nodes by~$T$.
          By
          line~\ref{sandglass:roundNumberMax} and
          line~\ref{sandglass:enterNewRound}, then, \emph{all} good
          nodes (including~$p_g$) must be in round~$(r-1)$ at~$T$.
          This contradicts our assumption and completes the proof.
\end{proof}

\newtheorem*{lemma11}{Lemma 11 in Sandglass~\cite{cryptoeprint:2022/796}}
\begin{lemma11} 
	\label{lemma:noGoodMsgthenFallBehind}	
	Suppose a good node~$p_g$  is in round~$r$ at step~$T$, and a node~$p_d$ is in round~$r_d$ at step~$T' \le T$.
	If~$p_d$ does not collect any messages from good nodes in any round~$(r-i)$, where~$0 \le i <k \varR$, then~$r_d \le (r-(k-1))$.
\end{lemma11}

\begin{proof}
	
	To prove the corollary, we compute the maximum number of messages~$D_{max}$
	that a defective node~$p_d$ can collect during the~$k\varR$
	rounds when it does not collect any message from good nodes. 
	%ittay
	%To help us count these messages, we let ~$T_{(r-k \varR+i)}$, where~$1 \le i \le k \varR$, denote, respectively, the earliest step for which all good nodes are at least in round~$(r-k \varR+i)$. 
	To help us count these messages, for any~$1 \le i \le k \varR$, denote by~$T_{(r-k \varR+i)}$ the earliest step for which all good nodes are at least in round~$(r-k \varR+i)$. 
	%By Corollary~\ref{corollary:onegoodmessageforallgood}, we know that~$T_r$ exists for all~$r$.
	%By Lemma~\ref{lemma:untimelyOneRoundAhead}, the latest round $p_d$ can be in at $T_{(r-k \varR+1)}$ is round $(r - k\varR + 2)$.
	
	Recall that, to be collected by~$p_d$ at step~$T'$, a message
	 from a good node must have been generated no later than step~$(T'-1) \le (T-1)$, and a message from a defective node
	must have been generated no later than step~$T' \le T$. Then, we partition
	the execution of the system up to step~$T$ into three time
	intervals, and compute, for each interval, the maximum number of
	messages generated during these intervals that~$p_d$ could have collected for rounds~$(r - k\varR
	+ 2)$ or larger.

	\begin{itemize}
		\item[I1:] Up to step~$(T_{(r-k \varR+1)} -1)$.
		
		By definition of~$T_{(r-k \varR+1)}$, some good node is in some round~$r' < r-k \varR+1$ at step~$(T_{(r-k \varR+1)}-1)$.
		Therefore, neither a defective node nor a good node can be in some
		round~$r'' >r-k \varR+1$ at step~$(T_{(r-k \varR+1)}-1)$,
		respectively because of
		Lemma~5
            %~\ref{lemma:untimelyOneRoundAhead} 
            and
		Corollary~2 in Sandglass~\cite{cryptoeprint:2022/796}
            %\ref{corollary:goodNodeOneroundApart}
            . Therefore,
		during this interval, no messages were generated for rounds~$(r-k
		\varR+2)$ or larger.
		
		\item[I2:] From~$T_{(r-k \varR+1)}$ up to~$(T_r-1$).
		
		By assumption,~$p_d$ only collects messages generated by defective
		nodes throughout interval I2. We further split I2
		into~$i$ consecutive subintervals, each going
		from~$T_{(r-k \varR+i)}$ up to~$(T_{(r-k\varR+i+1)}-1)$ for~$1 \le i \le (k \varR-1)$. By
		Lemma~10 in Sandglass~\cite{cryptoeprint:2022/796}
            %\ref{lemma:moreEffectiveMessage}
            , in each of these
		sub-intervals defective nodes can generate at most~$(\varR-1)$
		messages. Therefore, the number of messages generated by defective
		nodes during I2 is at most~$(\varR-1)\cdot (k \varR-1)$.
		
		\item[I3:]  From~$T_r$ to~$T$.
		
		Once again, by assumption,~$p_d$ only collects messages generated by defective
		nodes throughout interval I3.
            There are two cases:
		\begin{itemize}
			\item $T$ precedes~$T_r$.
			
			If so, defective nodes trivially generate no messages during I3.
			
			\item $T$ does not precede~$T_r$.
			
			By assumption, some good node~$p_g$ is in round~$r$ at~$T$, where
			it collects all messages generated by good nodes before~$T$;
			further, since~$p_g$ is still in round~$r$, the messages for
			round~$r$ sent by good nodes before~$T$ must be fewer
			than~$\varR$.
                Finally, since ~$p_g$ is in round~$r$ at~$T$, by
			Lemma~3 in Sandglass~\cite{cryptoeprint:2022/796}
                %\ref{lemma:goodNodeOneDeltaApart}
                , in all steps preceding~$T$
			no good node can be in round~$(r+1)$ or higher.
                We then conclude that from
			step~$T_r$ and up to~$(T-1)$ good nodes generated at
			most~$(\varR-1)$ messages, all for round $r$. 
			Let
			the number of messages generated by good nodes and defective nodes
			starting from~$T_r$ to~$T$ be, respectively,~$g_1,...,g_{k}$
			and~$d_1, ..., d_{k}$. 
			Then we have~$\Sigma_{i=1}^{k-1} g_i< \varR$.
			By Lemma~1 in Sandglass~\cite{cryptoeprint:2022/796}
                %\ref{lemma:selectR}
                , we then have~$\Sigma_{i=1}^{k} d_i< \varR$, \emph{i.e.}, during I3 defective
			nodes generate fewer than~$(\varR-1)$ messages.
		\end{itemize}
	\end{itemize}
	Adding the messages generated in the three intervals, we find  that~$D_{max}$,
	the maximum number of messages that~$p_d$ could have
	collected up to and including step~$T$ for rounds~$(r-k \varR+2)$ or larger, is
	smaller than~$(\varR-1)\cdot k \varR$; at the same time, since by
	assumption~$p_d$ is in round~$r_d$,~$D_{max}$ must equal at least~$(r_d - (r-k \varR+2)) \cdot \varR$.
	Therefore, we have that~$(r_d - (r-k \varR+2)) \cdot \varR <
	(\varR-1)\cdot k \varR$,
	which implies~$r_d \le r-(k-1)$, proving the corollary.
	
\end{proof}

Finally, we have our intended theorem.

\theoremSM*
Proof for the theorem directly follows from Lemma~2,~14,~15 in Sandglass~\cite{cryptoeprint:2022/796}.

%% file: proof.tex
\section{Gorilla Correctness}\label{sec:gorillaproof}

To prove the correctness of Gorilla, we first show a mapping in two steps from a Gorilla execution to an execution of Sandglass in~SM+ (\S\ref{sec:appendix-scaffolding}). Then, given that Sandglass is correct in SM+, we leverage this mapping  to proof safety (\S\ref{sec:appendix-safety}) and liveness (\S\ref{sec:appendix-liveness}) for Gorilla.

\subsection{Scaffolding}\label{sec:appendix-scaffolding}

\textForProofbeforeFirstDef

\defeqGGP*
\lemmareorg*
\begin{proof}
Consider any \GSand execution~$\eta_G$. 
We are going to construct an execution~$\eta_G^+$ in~\GMPLUS that satisfies \textsc{Reorg}-1,2,3.

First, we specify how correct and~Byzantine nodes join and leave in~$\eta_G^+$. 
For each correct node~$p$ in~$\eta_G$, a corresponding correct node~$p^+$ in~$\eta_G^+$ joins and leaves the system at the same steps as~$p$ in~$\eta_G$.
Consider any step~$s$ in~$\eta_G^+$, and let~$c$ be the number of correct nodes in step~$s$.
We make $(c-1)$~Byzantine nodes join at the beginning of step~$s$ and leave at the end of step~$s$.
	
 Let the set of valid messages sent by Byzantine nodes in~$\eta_G$ be~$\mathcal{M}_B$.
	Note that this set of valid messages sent by Byzantine nodes can be larger than the set of valid messages correct nodes received from Byzantine nodes due to Byzantine omissions.
	
	Our proof proceeds in two steps. We overview them and then explain them in detail:
	\begin{description}
		\item[Step 1] 
		For any~$m \in \mathcal{M}_B$, we assign a unique \textit{shell}, $(s,b)$, identified by a step~$s$ and a Byzantine node~$b$ in~$\eta_G^+$, for the $K$ Get() calls of VDF calculation for~$m$.
		Note that any node can only make one Get() call in a tick, and it takes~$K$ Get() calls to get $\textit{vdf}_m$.
		
		We prove four claims about the shells, which are useful later to prove the same messages can be generated in~$\eta_G$ and in~$\eta_G^+$.
		
		\item[Step 2] We prove by induction that correct nodes will receive and send the same messages in~$\eta_G^+$ as in~$\eta_G$, and the same valid messages are sent by~Byzantine nodes at the same step or earlier. 
        Then, it immediately follows that \textsc{Reorg}-1,2,3 are satisfied.
	\end{description}
	
	\noindent
	\textbf{\large Step 1}
	
	We run Algorithm~\ref{algo:reorg} to assign shells for the VDF calculation of messages in~$\mathcal{M}_B$.
        The algorithm operates as follows: we maintain two variables within a loop,~$s$ and CandidateVDF, where~$s$ denotes the step number, initially set to~$-1$, and CandidateVDF represents a set of VDFs, starting as an empty set.
            During each loop iteration,~$s$ is incremented by~$1$. 
            The algorithm adds VDFs whose first units are calculated in step~$s$ of~$\eta_G$ to the CandidateVDF set.
            While there exists an available shell in step~$s$, the algorithm assigns this shell to a VDF,~$vdf$, from CandidateVDF. The selected~$vdf$'s last unit should be calculated the earliest in $\eta_G$ and the algorithm then removes it from CandidateVDF.
            This process continues until no free shells remain in step~$s$. Subsequently, the algorithm moves to the next iteration of the loop and repeats these steps.
	
	\begin{algorithm}[t]
		\floatname{algorithm}{Algorithm}
		\caption{Algorithm for reorganizing VDF units}
		\label{algo:reorg}
		
		\begin{algorithmic}[1]
			\Procedure{Reorg}{} 
			
			\State $s \gets -1$
			\State CandidateVDFs$\gets \emptyset$
			
			\Loop
			\State $s \gets s+1; B_s \gets$ the set of Byzantine nodes at step~$s$ \label{algo:reorg:step}
			\State \textit{CandidateVDFs}~$\gets$ \textit{CandidateVDFs}~$\cup$ \{\textit{vdf} | \textit{vdf}'s first unit is calculated in step~$s$ in~$\eta_G$\}
			\label{algo:reorg:CandidateVDFs}
			
			\While{there's a free shell~$(s,p \in  B_s)$}
			\State \textit{vdf}$\gets$ a VDF result in \textit{CandidateVDFs} whose last unit is calculated the earliest in~$\eta_G$ 
			\label{algo:reorg:vdf}
			\State Assign~$(s,p)$ to~\textit{vdf}
			\label{algo:reorg:shell}
			\State \textit{CandidateVDFs}~$\gets$ \textit{CandidateVDFs}$\setminus$\{\textit{vdf}\}
			\label{algo:reorg:removevdf}
			\EndWhile
			
			\EndLoop
			
			\EndProcedure
			
		\end{algorithmic}
		
	\end{algorithm}
	
	Now we prove the following claims are true about the assignment:
	
	\begin{claim}
		\label{clm:startlate}
		For any~$m \in \mathcal{M}_B$, if the first Get() call for~$\textit{vdf}_m$ is in step~$i$, then the shell assigned to~$\textit{vdf}_m$ in~$\eta_G^+$ is in step~$i$ or later.
	\end{claim}
	
	\begin{proof}
		Since~$\textit{vdf}_m$ is not added to the \textit{CandidateVDFs} set at line~\ref{algo:reorg:CandidateVDFs} of Algorithm~\ref{algo:reorg} until~$s$ is increased to~$i$ (in line~\ref{algo:reorg:step}), the step of the shell that~$\textit{vdf}_m$ can be assigned to is at least~$i$ (line~\ref{algo:reorg:shell}).
	\end{proof}
 
		Before stepping into Claim~\ref{clm:endearly}, we show a useful observation following from Algorithm~\ref{algo:reorg}.
		\begin{observation}
			\label{obv:freeshell}
			Consider any step~$s$ in~$\eta_G^+$.
			We note two possible scenarios of shells at step~$s$, such that if either of these scenarios happens, the calculation of VDFs assigned to shells later than~$s$ in~$\eta_G^+$ must start in a step later than~$s$ in~$\eta_G$: 
			\subparagraph{A free shell exists at step~$s$ in~$\eta_G^+$.}  When the loop for~$s$ finishes, if there is a free shell at~$(s,p)$, then the CandidateVDFs set is empty at the end of the iteration for~$s$. That is, any \textit{vdf} assigned to a shell in a later step starts its calculation at a step~$s' > s$ in~$\eta_G$.
				
				\subparagraph{A shell at step~$s$ is assigned to a \textit{vdf} that is not the earliest to finish in~$\eta_G^+$}, among all the \textit{vdf}s that are not assigned to a shell yet. Consider the scenario in which a~\textit{vdf} is assigned to a shell at~$s$ at line~\ref{algo:reorg:shell} (it is the earliest to finish among all VDFs in CandidateVDFs), but it is not the earliest to finish among all the remaining VDFs. 
				Then the calculation of the remaining VDFs, including those VDFs whose calculation is finished earlier than~$vdf$, must start later than~$s$ in~$\eta_G$, because they are not in the CandidateVDFs set yet. 
			
		\end{observation}

    We are ready to prove Claim~\ref{clm:endearly}.
	
	\begin{claim}
		\label{clm:endearly}
		For any~$m \in \mathcal{M}_B$, if the last Get() call for~$\textit{vdf}_m$ is in step~$i$, then the shell assigned to~$\textit{vdf}_m$ in~$\eta_G^+$ is in step~$i$ or earlier.
	\end{claim}
	
	\begin{proof}		
		 We prove this by contradiction.
        Consider the first step~$i$ in~$\eta_G$, such that the last Get() call of a VDF,~$\textit{vdf}^{\ *}$, is in step~$i$, but~$\textit{vdf}^{\ *}$ is not assigned to a shell in~$\eta_G^+$ by the end of $i$-th iteration of the loop.
		
		Consider the largest step~$j \le i$ such that, one of the scenarios in Observation~\ref{obv:freeshell} happens. 
		If no such~$j$ exists, we take~$j$ to be~-1.
		By Observation~\ref{obv:freeshell}, for all the VDFs that are already assigned to the shells in steps~$[j+1,i]$, their first Get() calls are after step~$j$.
		Furthermore, since they are already assigned to shells, their last unit is no later than~$\textit{vdf}^{\ *}$'s, \emph{i.e.}, their last Get() calls are in or before step~$i$ .
		We call this set of \textit{vdf}s,~$VDF_{occupy}$.
		Then, we have all of the Get() calls of~$VDF_{occupy}$ are in steps~$[j+1,i]$ in~$\eta_G$.
		
		Note that there are no free shells in~$[j+1,i]$ (otherwise,~$j$ would be larger).
		Let the number of~Byzantine nodes in any step~$s$ be~$b_s$.
		Then, the size of~$VDF_{occupy}$ is~$\Sigma_{s=j+1}^i b_s$.
		Therefore the number of Get() calls for~$VDF_{occupy}$ in steps~$[j+1,i]$ is~$K \cdot\Sigma_{s=j+1}^i b_s$ in~$\eta_G$.
		
		Note that in~$\eta_G$, the number of Byzantine nodes at any step~$s$ is at most~$b_s$. 
        Therefore, the total number of Get() calls that can be made in steps~$[j+1,i]$ in~$\eta_G$ is at most~$K \cdot\Sigma_{s=j+1}^i b_s$, and one of them is the last Get() call of~$\textit{vdf}^{\ *}$.
		Therefore, there are not enough ticks available to make all Get() calls for~$VDF_{occupy}$ in~$\eta_G$.
		A contradiction.
	\end{proof}
	
	\begin{claim}
		\label{clm:relativeByzantine}
		Consider any two VDFs,~$\textit{vdf}_1$ and~$\textit{vdf}_2$, reserved respectively at steps~$s_1$ and~$s_2$ in~$\eta_G^+$.
		If the last Get() call of~$\textit{vdf}_1$ is before the first Get() call of~$\textit{vdf}_2$ in~$\eta_G$, then~$s_1 \le s_2$.
	\end{claim}
	
	\begin{proof}
		By line~\ref{algo:reorg:CandidateVDFs},~$\textit{vdf}_2$ must be added to the CandidateVDFs set after~$\textit{vdf}_1$. 
		By line~\ref{algo:reorg:vdf},~$\textit{vdf}_1$ must be assigned to its shell before~$\textit{vdf}_2$.
		Note that, by line~\ref{algo:reorg:step}, shells are assigned in non-decreasing step order; therefore,~$s_1 \le s_2$.
	\end{proof}

	\begin{claim}
		\label{clm:norecursive}
		Consider any three \textit{vdf}s,~$\textit{vdf}_1$,~$\textit{vdf}_2$, and~$\textit{vdf}_3$, reserved respectively at steps~$s_1$,~$s_2$ and~$s_3$ in~$\eta_G^+$.
		If, in~$\eta_G$, the last Get() call of~$\textit{vdf}_1$ is before the first Get() call of~$\textit{vdf}_2$, and the last Get() call of~$\textit{vdf}_2$ is before the first Get() call of~$\textit{vdf}_3$, then~$s_1 < s_3$.
	\end{claim}
	
	\begin{proof}
		Let the tick of the first Get() call of~$\textit{vdf}_2$, and~$\textit{vdf}_3$ in~$\eta_G$ be~$t_2^f$, and~$t_3^f$.
		
		Note that the last Get() call of~$\textit{vdf}_1$ is before~$t_2^f$ in~$\eta_G$. 
            By Claim~\ref{clm:endearly}, ~$\textit{vdf}_1$ must be assigned to step ~$\lfloor(t_2^f-1)/K\rfloor$ or earlier.
		By Claim~\ref{clm:startlate},~$\textit{vdf}_3$ must be assigned to~$\lfloor t_3^f/K\rfloor$ or later.
		Since $ t_2^f  \le t_3^f - K$, \emph{i.e.},  $( t_2^f-1)  < t_3^f - K$, we have ~$s_1 \le \lfloor(t_2^f-1)/K\rfloor < \lfloor t_3^f/K\rfloor \le s_3$, \emph{i.e.},~$s_1 < s_3$.

	\end{proof}

	Now in~$\eta_G^+$, all the~Byzantine nodes join and leave at the boundaries and stay for a single step; and for each valid message~$m$ sent by a Byzantine node in~$\eta_G$, we have assigned a unique shell~$(s,b)$ for some~$s$ and~$b^*$ to it.
	Then, if~$b$ can receive the messages contained in~$m$'s message coffer,~$M_m$, then~$b$ will make~$K$ Get() calls in step~$s$ for input~$(M_m, nonce_m)$, and therefore~$b$ will be able to send~$m$ in~$\eta_G^+$.
 
	We construct~$\eta_G^+$ so that if a Byzantine node is able to send a (valid) message~$m$, it sends~$m$ to all Byzantine nodes in the next step.
	Every Byzantine node forwards all the messages it has received in a step to all the Byzantine nodes in the next step.
	Furthermore, if~$m$ is received at tick~$t'$ by a correct node~$c$ in~$\eta_G$, one Byzantine node who has~$m$ at tick~$(t'-1)$ will send~$m$ to~$c$ at tick~$(t'-1)$ in~$\eta_G^+$.
	
	We will show in Step~2 that for each valid message~$m$ sent by a Byzantine node in~$\eta_G$ and the shell~$(s,b)$ assigned to it,~$b$ can indeed receive the messages contained in~$m$'s message coffer, and can therefore send~$m$ at~$s$.
	Furthermore, any correct node~$c^+$ at step~$s$ in~$\eta_G^+$ can receive the same set of messages as its corresponding node~$c$ in~$\eta_G$ at step~$s$, and therefore send the same message at~$s$.
	\newline

	\noindent
	\textbf{Step 2}
	
	%We will prove by induction that correct nodes will send the same messages in~$\eta_G^+$ as in~$\eta_G$, and the same valid messages will be sent by Byzantine nodes at the same step or earlier.
	We will prove by induction that we can construct~$\eta_G^+$ such that any message received and sent in~$\eta_G$ by correct nodes is received and sent in~$\eta_G^+$ at the same step, and any valid message sent in~$\eta_G$ by a Byzantine node is sent in~$\eta_G^+$ at the same or earlier step by a Byzantine node and is received at the same step as in~$\eta_G$ by correct nodes.
	
	Recall that in \GMPLUS, Byzantine nodes can peek VDF results of other Byzantine nodes that complete at the same step.
	
	\begin{description}
		\item[Base case]
		In step~0, we show it is possible to ($i$) make all correct nodes in~$\eta_G^+$ send the same messages sent by their corresponding correct nodes in~$\eta_G$ at step~0, and ($ii$) for any shell~$(0,p_{sh})$ that is reserved for message~$m_{sh}$, make~$p_{sh}$ send~$m_{sh}$.
		
		First, consider the valid messages that are sent by correct nodes in step~0 in~$\eta_G$.
		In step~0, there were not enough ticks to generate a VDF. Therefore, all correct nodes receive no (valid) message in~$\eta_G$, \textit{i.e.}, the message coffer of any valid message sent by a correct node in step~0 in~$\eta_G$ is empty.
		Therefore, by making all the correct nodes have the same initial values as in~$\eta_G$ and pick the same nonces, any (valid) message sent in~$\eta_G$ can also be sent in~$\eta_G^+$ in step~0.
		
		Second, consider any message~$m_{sh}$ whose shell is in step~0 in~$\eta_G^+$.
		Consider the message coffer~$M_{sh}$ of~$m_{sh}$.
		Note that in~$\eta_G$, for any~$m \in M_{sh}$,~$m$'s last Get() call must be before~$m_{sh}$'s first Get() call.
		Therefore, by Claim~\ref{clm:relativeByzantine},~$m$'s shell must also be in step~0.
		Since step~0 is the earliest step, by Claim~\ref{clm:norecursive},~$m$ does not contain any messages in its message coffer, otherwise messages in~$m$'s message coffer would have been assigned to a step earlier than step~0.
		Then,~$p^+$ can include the messages in~$M_{sh}$ in its message coffer in~$\eta_G^+$ by peeking the~$\textit{vdf}_m$, therefore,~$p^+$ can send~$m_{sh}$ in step~0.
		Again, by picking the same nonce,~$m_{sh}$ can be sent in~$\eta_G^+$ in step~0.
		
		We make the nodes send these messages in~$\eta_G^+$ in the following way:
		\begin{itemize}
			\item Correct nodes send their messages to all the nodes.
			\item Byzantine nodes send their messages to all the Byzantine nodes.
			\item Consider a correct node~$p_c$ that receives a message~$m$ sent by a Byzantine node in~$\eta_G$ at step~1.
			Note that by Claim~\ref{clm:endearly},~$m$ must be assigned to a shell in step~0.
			Let that shell be~$(0,b_m)$.
			$b_m$ sends~$m$ to~$p_c$ at step~0 in~$\eta_G^+$.
			%Note that since the network is synchronous, if a correct node~$p_c$ receives a message~$m$ from a Byzantine node at step~1, there must exist a Byzantine node at step~0.
		\end{itemize}  
		
		\item[Induction hypothesis] 
		Up until step~$k$, any message sent in~$\eta_G$ by a correct node is sent in~$\eta_G^+$ at the same step.
		Byzantine nodes can send all the messages whose shells are at step~$K$.
		Messages received by correct nodes up until step~$(k+1)$ are the same as in~$\eta_G$.
		Byzantine nodes receive all the messages from correct nodes and Byzantine nodes.
	
		\item[Induction step] 
		Consider messages sent at step~$(k+1)$ in~$\eta_G$.
		
		%By the induction hypothesis, until step~$k$, all the messages that are sent by correct nodes in~$\eta_G$ are also sent in~$\eta_G^+$.
		
		First, we prove correct nodes can send the same messages in step~$(k+1)$ in~$\eta_G^+$ as in~$\eta_G$.
		By the induction hypothesis and due to synchrony, correct nodes receive the same set of messages from Byzantine nodes and correct nodes in step~$(k+1)$ in~$\eta_G^+$ and in~$\eta_G$.
		By making all the correct nodes select the same nonces in~$\eta_G^+$ as in~$\eta_G$, correct nodes are going to send the same messages in~$\eta_G^+$ as in~$\eta_G$.
		
		Second, we prove that Byzantine nodes can send all the messages whose shells are at step~$(k+1)$ in~$\eta_G^+$. 
		Consider any message~$m_{sh}$ whose shell is at step~$(k+1)$ and any~$m_M$ in~$m_{sh}$'s message coffer in~$\eta_G$.
        There are two possibilities for~$m_{M}$.
		\begin{description}
			\item[$m_M$ is sent by a Byzantine node]
			Then the last Get() call for the VDF of~$m_M$ is before the first Get() call for the VDF of~$m_{sh}$ in~$\eta_G$.
			Then, by Claim~\ref{clm:relativeByzantine}, the shell reserved for~$m_M$ is in a step~$s_M$, where $s_M \le (k+1)$.
			Note that in~$\eta_G^+$, Byzantine nodes in a step can peek at the VDF results for messages from Byzantine nodes that  finish within the same step; and, by Claim~\ref{clm:norecursive},~$m_M$ does not contain any message whose shell is also in step~$(k+1)$ or later.
			It follows that~$m_M$ can be included in the message coffer of~$m_{sh}$ in~$\eta_G^+$.
			
			\item[$m_M$ is sent by a correct node]
			By Claim~\ref{clm:startlate},~$m_{sh}$ is calculated no later than step~$(k+1)$ in~$\eta_G$.
			Therefore,~$m_M$ must be sent in a step no later than step~$k$ in~$\eta_G$.
			By the induction hypothesis,~$m_M$ is sent in~$\eta_G^+$.
			Therefore,~$m_M$ can be included in~$m_{sh}$'s coffer in~$\eta_G^+$.
		\end{description}
		Therefore, by picking the same nonces, Byzantine nodes can also send the same messages as in~$\eta_G$ at step~$(k+1)$ in~$\eta_G^+$ whose shells are at step~$(k+1)$.
		\newline

		We make the nodes send messages in~$\eta_G^+$ in the following way:
		\begin{itemize}
			\item Correct nodes send their messages to all the nodes.
			\item Byzantine nodes send  their messages to all the Byzantine nodes.
                \item Byzantine nodes forward messages they received to all the Byzantine nodes.
			\item Consider a correct node~$p_g$ that receives a message~$m$ sent by a Byzantine node in~$\eta_G$ at step~$(k+2)$.
			By Claim~\ref{clm:endearly},~$m$ must be assigned to a shell that is no later than~$(k+1)$.
			We just showed above that Byzantine nodes can send all the messages whose shells are at step~$(k+1)$.
			Combining with induction hypothesis, some Byzantine node $b_m$ must have received or generated~$m$. 
			We make~$b_m$ sends~$m$ to~$p_g$ at step~$(k+1)$ in~$\eta_G^+$.
				%Note that since the network is synchronous, if a correct node~$p_g$ receives a message~$m$ from a Byzantine node at step~$(k+2)$, there must exist a Byzantine node at step~$(k+1)$.
		\end{itemize}  
%		
%		%Then, we will prove for any message~$m$ sent by a Byzantine node at step~$(k+1)$ in~$\eta_G$, a Byzantine node sends the same messages in step~$(k+1)$ or earlier.
%		Now consider the messages correct nodes received from Byzantine nodes at step~$(k+2)$ in~$\eta^G$.
%	
%		%Combining with the induction hypothesis,~$m$ must be sent at step~$(k+1)$ or earlier.
%		
%		We make the nodes send these messages in~$\eta_G^+$.
%		Specifically, correct nodes send messages to all the nodes; Byzantine nodes send and forward messages to all the Byzantine nodes; if a correct node~$p_g$ receives a message~$m$ sent by a Byzantine node in~$\eta_G$ at step~$(k+2)$, then one Byzantine node send~$m$ to~$p_g$ at step~$(k+1)$ in~$\eta_G^+$.
%	
		
	\end{description}
Now we have proven that correct nodes receive and send the same messages in~$\eta_G$ as in~$\eta_G^+$.
Note that we also proved that Byzantine nodes can send all the messages at the step where their shells are. 
By Claim~\ref{clm:endearly}, it directly follows that for any valid message~$m$ sent by a Byzantine node in~$\eta_G$, a Byzantine node sends the same messages in the same step or earlier in~$\eta_G^+$ (\textsc{Reorg}-3). 
 
In summary, we have constructed~$\eta_G^+$ in~\GMPLUS that satisfies \textsc{Reorg}-1,2,3.

\end{proof}

\textBetweenlemmareorgAnddefMAP

\defMAP*
Thus, we can define the execution mapping:
\defeqLVS*
\lemmainterpret*
\begin{proof}
	For execution~$\eta_G^+$, we construct the interpretation of~$\eta_G^+$,~$\eta_S^+$, in~SM+.
    First, for every $p$ in $\eta_G^+$, we add a corresponding $\hat{p}$ to $\eta_S^+$ such that:
    \begin{itemize}
        \item  $\hat{p}$ joins and leaves at the same steps that $p$ joins and leaves, respectively.
	\item $\hat{p}$ has the same initial value as $p$.
        \item If $p$ is a Byzantine node, then $\hat{p}$ is a defective node; otherwise, $\hat{p}$ is a good node.
    \end{itemize}  

    The number of Byzantine nodes in a step of $\eta_G^+$ is smaller than the number of correct nodes at each step and the number of defective and good nodes are equal to those of Byzantine and correct nodes, respectively; 
    therefore, the number of defective nodes in~$\eta_S^+$ is fewer than that of good nodes at each step.
    Thus, Condition~\ref{def:equi:1},~\ref{def:equi:2} and~\ref{def:equi:3} are satisfied.
	
    We now construct the messages sent in~$\eta_S^+$ such that~$\eta_S^+$ is an interpretation of~$\eta_G^+$. 
    Specifically, we require Condition~\ref{def:equi:4} and Condition~\ref{def:equi:5} in Definition~\ref{def:eq-L-VS} to be satisfied for the messages sent.
    We prove this by induction on steps.
    Note that messages are constructed inductively alongside the induction.
    
    \begin{description}%[listparindent = 1.5em]
        \item[Induction Base] 
        Consider any node~$p$ in~$\eta_G^+$ at the first step, and message~$m$ that it generates at the first step.
        We prove the claim holds for the first step, conditioned on whether~$p$ is a correct node or not.
        \begin{description}
            \item[$p$ is a correct node]  
        Note that~$p$ and~$\hat{p}$ have the same initial value~$v_p$ by construction.
        We now prove that~$\hat{p}$ will send~$\hat{m} = \textsc{Mapm}(m)$ to all nodes at the first step in~$\eta_S^+$.

        Since there are not enough ticks to generate a VDF,~$p$ does not receive any message in the first step, therefore,~$m$ does not contain any message in its message coffer.
        Therefore,~$m = (p,r=1,v=v_p,\varpriority =0,\varpriorityCounter=0,M=\emptyset,\cdot, \cdot)$.
        Note that~$\hat{p}$ cannot receive any message in the first step, either. 
        The message~$\hat{p}$ sends at the first step, by Sandglass, is~$(\hat{p},uid=1,r=1,v=v_p,\varpriority =0,\varpriorityCounter=0,M=\emptyset)$, which is equal to~$\hat{m}$.
        Note that in~$\eta_S^+$, since good nodes are synchronously connected, all the good nodes in the second step will receive~$\hat{m}$.
        Note that it's possible for a defective node not to receive~$\hat{m}$, by performing an omission failure, or to receive~$\hat{m}$ in any step, being asynchronously connected to other nodes.
        %We will specify when and which defective nodes receive~$\hat{m}$ later.

        \item[$p$ is a Byzantine node]
        Note that in~SM+, defective nodes at step~$s$ can receive any message~$m$ sent at~$s$, as long as~$m$ does not contain in its coffer a message that is also sent at~$s$. 
        Again, we will prove~$\hat{m}$ can and will be sent at the first step in~$\eta_S^+$.

        We first consider any message~$m$ generated by a Byzantine node~$p$, whose message coffer is empty.
        With the same argument for a correct node~$p$ in the first step of~$\eta_G^+$,~$\hat{p}$ will send~$\hat{m}$ in the first step of~$\eta_S^+$.
        Again, note that it's possible for any node not to receive~$\hat{m}$, when~$\hat{p}$ performs an omission failure, or to receive~$\hat{m}$ in any step, since~$\hat{p}$ is asynchronously connected to other nodes.
        %We specify when and which nodes receive~$\hat{m}$ later.

        Second, let's consider a message~$m$ generated by a Byzantine node~$p$, whose message coffer is not empty.
        Note that the number of messages that can be included in~$m$'s coffer is at most the number of Byzantine nodes in the first step, which is smaller than~\varR.
        We have~$m = (p,r=1,v,\varpriority =0,\varpriorityCounter=0,M,\cdot, \cdot)$.
        Consider any message~$m_{pk} \in M$,~$m_{pk}$ must be sent by a Byzantine node and the \textit{vdf} in~$m_{pk}$ is a Peek() result.
        By specification of Peek(),~$m_{pk}$ must be generated in the first step, and not contain any message in its message coffer.
        We have proven above that~$\hat{m}_{pk}$ will be sent in the first step of~$\eta_S^+$. 
        Note that message coffer of~$\hat{m}_{pk}$ is also empty.
        Therefore, we make the scheduler deliver~$\hat{m}_{pk}$ to the defective node~$\hat{p}$ and therefore~$\hat{p}$ will include~$\hat{m}_{pk}$ in its message coffer.
        In summary, for every~$m_{pk} \in M$,~$\hat{p}$ will receive~$\hat{m}_{pk}$, and~$|M| < \varR$.
        Therefore,~$\hat{p}$ will send~$ (p,r=1,v,\varpriority =0,\varpriorityCounter=0,\textsc{Mapm}(M))$, which is equal to~$\hat{m}$.

        \end{description}
        
        \item[Induction Hypothesis] 
        Node~$\hat{p}$ receives~$\hat{m}$ at step~$s' \le s$ in~$\eta_S^+$, iff~$p$ receives~$m$ at step~$s'$ in~$\eta_G^+$ (Condition~\ref{def:equi:4}).
        Node~$\hat{p}$ sends~$\hat{m}$ at step~$s' \le s$ in~$\eta_S^+$, iff~$p$ generates a message~$m$ in~$\eta_G^+$ at step~$s'$ (Condition~\ref{def:equi:5}). 
        
        \item[Induction Step] Now we prove that the claim holds for step $s+1$, conditioned on whether~$p$ is a correct node or not.

        \begin{description}
        \item[$p$ is a correct node] 
        We prove Condition~\ref{def:equi:4} and Condition~\ref{def:equi:5} separately.
        \begin{description}
            \item[Condition~\ref{def:equi:4}]
            First, we will prove that if~$p$ receives~$m$ at step~$(s+1)$ in~$\eta_G^+$, then~$\hat{p}$ receives~$\hat{m}$ at step~$(s+1)$ in~$\eta_S^+$.
        
        If~$m$ is generated by a correct node~$p_c$,~$m$ must be sent at step~$s$ in~$\eta_G^+$.
        By the induction hypothesis,~$\hat{m}$ must be sent by a good node~$\hat{p_c}$ at step~$s$ in~$\eta_S^+$.
        Note that in~SM+, the network between good nodes is synchronous.
        Therefore,~$\hat{p}$ receives~$\hat{m}$ at~$(s+1)$ in~$\eta_S^+$.

        If~$m$ is generated by a Byzantine node~$p$,~$m$ could have been generated at any step~$s' \le s$ and finally sent to~$p$ at~$s$ in~$\eta_G^+$.
        By the induction hypothesis,~$\hat{m}$ must be sent by a defective node~$\hat{p}$ at step~$s'$ in~$\eta_S^+$.
        Note that in~SM+, the network between good nodes and defective nodes is asynchronous.
        Therefore, we can make the scheduler deliver~$\hat{m}$ to~$\hat{p}$ at~$(s+1)$ in~$\eta_S^+$.

        \item[Condition~\ref{def:equi:5}] 
         Now we prove that if~$p$ generates a message~$m$ in~$\eta_G^+$ at step~$(s+1)$, then~$\hat{p}$ can send~$\hat{m}$ at step~$(s+1)$ in~$\eta_S^+$.
        Consider the set of messages~$Rec_p$ received by~$p$ at~$s$.
        %Note that~$\eta_G^+$ is equivalent to~$\eta_G$. Therefore, all the messages sent are valid by Lemma~\ref{lemma:reorg}.
        By the induction hypothesis and the proof for Condition~\ref{def:equi:4},~$\hat{p}$ received a set of messages~$\textsc{Mapm}(Rec_p)$ at step~$(s+1)$.
        
        Note that Gorilla (lines \ref{gorilla:resetM}-\ref{gorilla:unionCollectedMessagesEnd},\ref{gorilla:unionSameRoundMsg}) and Sandglass (lines \ref{sandglass:resetM}-\ref{sandglass:unionCollectedMessagesEnd},\ref{sandglass:unionSameRoundMsg}) construct message coffers in the same way based on the messages received.
        Let the message coffer maintained by~$p$ be~$M_p$.
        Then, the message coffer~$\hat{p}$ has maintained up to~$(s+1)$ is~$\textsc{Mapm}(M_p)$.
        
        It follows that~$\hat{p}$'s set~$C_{\hat{p}}$ at line~\ref{sandglass:multiset} is~$\textsc{Mapm}(C_p)$.
        If the proposal value~$v_p$ of~$p$ is chosen based on the \textit{vdf} at line~\ref{gorilla:vdfValueUpdate}, then~$\hat{p}$ also chooses proposal value~$v_{\hat{p}}$ based on a random selection at line~\ref{sandglass:valueUpdateEnd}.
        Then in~$\eta_S^+$,~$\hat{p}$ chooses the value~$v_p$.
        Note that any~$m_M \in M_p$ and~$\textsc{Mapm}(m_M)$ have the same round number, proposal value, \varpriority and \varpriorityCounter.
        The variables \varpriority and \varpriorityCounter are updated the same way in Gorilla (lines~\ref{gorilla:setPriorityCounter}-\ref{gorilla:setPriority}) and in Sandglass (lines~\ref{sandglass:setPriorityCounter}-\ref{sandglass:setPriority}).
        Therefore,~$\hat{p}$ sends message~$\textsc{Mapm}(m)$ at~$s$, which has~$\hat{p}$ as the process, the same round number, proposal value, \varpriority and \varpriorityCounter as~$m$, and~$\textsc{Mapm}(M_p)$ as the message coffer.
        
        \end{description}

        \item[$p$ is a Byzantine node]

        We further separate this case into two sub-cases:
        \begin{description}
            \item[$p$ does not receive a peek result]
            First, we prove that if~$p$ receives~$m$ at step~$(s+1)$ in~$\eta_G^+$, then~$\hat{p}$ receives~$\hat{m}$ at step~$(s+1)$ in~$\eta_S^+$.
            Since~$m$ is not a Peek() result,~$m$ must be sent at step~$s$ or earlier in~$\eta_G^+$.
            By the induction hypothesis,~$\hat{m}$ must be sent at step~$s$ or earlier in~$\eta_S^+$.
            Note that in~SM+ connections to defective nodes are asynchronous.
            Therefore, we can make the scheduler deliver~$\hat{m}$ to~$\hat{p}$ at~$(s+1)$ in~$\eta_S^+$.

            Now we prove that if~$p$ generates a message~$m$ in~$\eta_G^+$ at step~${(s+1)}$, then~$\hat{p}$ can send~$\hat{m}$ at step~$(s+1)$ in~$\eta_S^+$.

            Consider~$m= (r_m,v_m,\varpriority_m, \varpriorityCounter_m, M_m, nonce_m, \textit{vdf}_m)$.
            Note that by Lemma~\ref{lemma:reorg}~$p$ joins the system for only one step.
            Therefore, for any~$m_M \in M_m$,~$p$ receives it at step~$(s+1)$.
            Therefore, from what we proved above,~$\hat{p}$ receives~$m_M$ at step~$(s+1)$.
            Therefore~$\hat{p}$ has~$\textsc{Mapm}(M_m)$ as its coffer.
            
            Now we prove that~$\hat{p}$ can have~$v_m$ as proposal value,~$\varpriority_m$ as priority, and~$\varpriorityCounter_m$ as \varpriorityCounter \ following the Sandglass protocol. 
            Note that by Lemma~\ref{lemma:reorg},~$m$ must be valid (isValid(m) equals true).

            \begin{description}
                \item[$\varpriorityCounter_m$ is~0] For all round~$(r_m-1)$ messages in~$M_m$, the largest priority value for proposal values~$a$ and~$b$ are the same.
                Therefore, in $\textsc{Mapm}(M_m)$, the largest priority value for proposal values~$a$ and~$b$ are the same.
                Then, it is possible that~$\hat{p}$ chooses~$v_m$ as its proposal value at line~\ref{sandglass:valueUpdateEnd}.
                \item[$\varpriorityCounter_m$ is not~0]
                All round~$(r_m-1)$ messages in~$M_m$ have the same proposal value~$v_m$.
                Therefore, in $\textsc{Mapm}(M_m)$, all round~$(r_m-1)$ messages also propose~$v_m$.
                By Sandglass protocol (line~\ref{sandglass:valueUpdateStart}),~$\hat{p}$ will set its proposal value to~$v_m$.
            \end{description}          
            Since~$m$ is consistent,~$\hat{p}$ will set \varpriorityCounter \  to $\varpriorityCounter_m$ and \varpriority to $\varpriority_m$.

            Note that~$\hat{p}$ joins the system for only one step, and therefore it sends only one message. 
            We have~$\hat{p}$ will send~$\hat{m} = (\hat{p},uid=1,r_m,v_m,\varpriority_m, \varpriorityCounter_m, \textsc{Mapm}(M_m))$.
            
            \item[$p$ receives some peek results]
            Consider any peek result~$m'$ that~$p$ receives at step~$(s+1)$.
            Due to the constraint on peeking,~$m'$ does not contain any peek result, and~$m'$ is sent by a Byzantine node within step~$(s+1)$.
            By what we proved in the first sub-case,~$\hat{m'}$ must be sent in step~$(s+1)$.
            Since in~SM+, the defective node~$\hat{p}$ can receive~$m'$.

            Then, based on the same argument as the first sub-case, we can show that~$\hat{p}$ can send~$\hat{m}$.
        \end{description}
        \end{description}
    \end{description}
\end{proof}

%Lemma~\ref{lem:Lock-to-VS} is proving the intuitive guess that, since the GS and SM+ protocols are extremely similar, it should be possible to directly map them to each other if the GS execution conforms to step boundaries.
%In fact, one can add more formality to the proof, and show that the GS protocol is actually a refinement~\cite{} of the VS protocol if it were to be used in the VS setting (hybrid synchrony and defective nodes).
%We do not mention the details, and the current form of the theorem is sufficient for our purposes.
%\begin{definition}[\textsc{Reorg}]
%	We denote the mapping of Lemma~\ref{lem:GS-to-lock} that maps an execution $\eta_{G}$ from the GS paradigm to an execution $\eta_L$ in the lock-step GS paradigm by $\textsc{Reorg}$, \emph{i.e.}, $\eta_{L} = \textsc{Reorg}(\eta_{G})$.
%\end{definition}
%\begin{definition}[\textsc{Interpret}]
%	We denote the mapping of Lemma~\ref{lem:Lock-to-VS} that maps an execution $\eta_L$ from the lock-step GS paradigm to an execution $\eta_S^+$ in the VS paradigm by $\textsc{Interpret}$, \emph{i.e.}, $\eta_S = \textsc{Interpret}(\eta_L)$.
%\end{definition}
\subsection{Safety}\label{sec:appendix-safety}
\textAtStartOfSafety

\lemmadecisionreorg*
\begin{proof}
    Consider any execution~$\eta_G$ in~GM and $\eta_G^+ = \textsc{Reorg}(\eta_G)$ in GM+.

    Consider a correct node $p$ that decides a value $v$ at step $s$ in $\eta_G$.
    Consider the message~$m= (r,v, \varpriority, \varpriorityCounter, M, nonce, \textit{vdf})$ that~$p$ sends right after it decides.
    Note that by Gorilla protocol, this is the first step that~$p$ ever collects at least~\varR messages for round~$(r-1)$ and $\varpriority \ge 6\varR +4$.
    By Lemma~\ref{lemma:reorg},~$p^+$ receives and sends the same messages as~$p$ in the same steps.
    Therefore,~$s$ is also the first step that~$p^+$ collects at least~\varR messages for round~$(r-1)$, and~$p^+$ sends~$m$ at step~$s$.
    Therefore,~$p^+$ must also have decided~$v$ at step~$s$.
\end{proof}

\lemmadecisioninterpret*
\begin{proof}
    Consider any execution~$\eta_G$ in~GM. 
	Consider execution~$\eta_G^+ = \textsc{Reorg}(\eta_G)$ in GM+, and~$\eta_S^+$ in~SM+ is an interpretation of~$\eta_G^+$.
	
	\begin{description}
		\item[Proof for Statement (\ref{lemma:equivalence:statement1}):] Consider a correct node $p$ that decides a value $v$ at step $s$ in $\eta_G^+$.
		Consider the message~$m= (r,v,priority, \varpriorityCounter, M, nonce, \textit{vdf})$ that~$p$ sends right after it decides.
		Note that by Gorilla protocol, this is the first step that~$p$ ever collects at least~\varR messages for round~$(r-1)$ and $priority \ge 6\varR +4$.
		Since~$\eta_S^+$ is an interpretation of~$\eta_G^+$, by Definition~\ref{def:eq-L-VS},~$\hat{p}$ receives the equivalent messages received by~$p$ in the same steps.
		Therefore,~$s$ is also the first step that~$\hat{p}$ collects at least~\varR messages for round~$(r-1)$.
		Note that by definition of equivalence,~$\hat{p}$ sends~$\hat{m}= (\hat{p}, \cdot, r,v,priority, \varpriorityCounter, M)$ at~$s$ with~$priority \ge 6\varR +4$.
		Then~$\hat{p}$ must also have decided~$v$ at step~$s$.
		
		\item[Proof for Statement (\ref{lemma:equivalence:statement2}):] 
		Consider the first message~$m = (r,v,priority=0, \varpriorityCounter=0, M, nonce, \textit{vdf})$ that~$p$ generates for round~$r$, and assume~$m$ is generated at step~$s$.
		By the definition of equivalence,~$\hat{p}$ sends~$\hat{m}= (\hat{p}, \cdot, r,v,priority=0, \varpriorityCounter=0, M)$ at~$s$ in~$\eta_S^+$, and~$\hat{m}$ is also the first message~$\hat{p}$ sent for round~$r$.
		Since~\varpriorityCounter is~0,~$v$ is chosen based on a coin toss by Sandglass protocol.
		Therefore,~$\hat{p}$ must have flipped a coin at step~$s$ and the result is~$v$.
	\end{description}
\end{proof}

\thmgorillaAgr*
\proofthmgorillaAgr*
\thmgorillaVad*
\proofthmgorillaVad*

\subsection{Liveness}\label{sec:appendix-liveness}
\textAtStartOfLiveness

\defMesHis*
\defCon*
\defEnvG*
\defStrategyB*

\textRandomizedStrategyAli
\lemmastrategyRD*
\begin{proof}
    Let us fix the environment~$\mathcal{E}$.
    Consider a randomized Byzantine strategy~$\Theta_{\mathcal{E}}$ that achieves a positive non-termination probability.
    We omit~$\mathcal{E}$ from the subscript for brevity.
    Denote the non-termination event in~$H_\Theta$ with $NT$, \emph{i.e.},~$P_{H_{\Theta}}(NT) > 0$.
    For brevity, we drop~$H_\Theta$ from the subscripts for the probabilities.

    We provide an inductive proof.
    Consider the first time that a Byzantine node takes a randomized action. 
    This node can only choose from a countable set of actions: sending a message to a node or calculating a unit of
    VDF.
    Let us call this set of actions~$\mathcal{A} = \{A_1, A_2, \dots\}$.
    We have $P(NT) = \sum_{i}P(A_i)P(NT|A_i)$.
    Now, since $P(NT) > 0$, there should exist an $A_i$ such that $P(NT|A_i) > 0$, \emph{i.e.}, the Byzantine nodes could have
    achieved positive non-termination probability by \emph{deterministically} taking the action $A_i$.
    Similarly, for every further randomized action with execution prefix~$\phi$, and with the action choices~$\mathcal{A}'=\{A_1', A_2',\dots\}$, if~$P(NT|\phi) > 0$, then there should exist an~$A_i'$ such that~$P(NT|\phi,A_i') > 0$.
    Repeating this process, we can carve a deterministic strategy from~$\Theta$, such that non-termination still has a positive probability.
\end{proof}
\textAfterDefEnv

\defTermProbOne*

\textAfterdefTermProbOne

\lemmareorgstrategy*
\prooflemmareorgstrategy*

We can now continue our proof by showing that the \textsc{Interpret} mapping preserves the non-zero non-termination probability,
which will help us prove our desired liveness property, termination with probability~1.
In order to do this, we first introduce a machinery that allows us to prove the perseverance the non-zero non-termination probability
throughout the \textsc{Interpret} mapping.
However, first note that the \textsc{Interpret} mapping changes nothing in the executions, and only maps the
``syntax'' of~GM+ to that of~SM+.
Therefore, the mapping does not introduce non-deterministic decision points for Byzantine nodes.
Furthermore, as shown in Lemma~\ref{lemma:interpret}, actions by Byzantine nodes in any execution in~GM+
are translated to actions by the network in~SM+.
Therefore, we conclude that a Byzantine strategy $\Theta$ in~GM+ is mapped to a network strategy
in~SM+.
Abusing notation, let us show the mapped strategy with $\textsc{Interpret}(\Theta)$.
\begin{definition}\label{def:events}
	Given strategy $\Theta$ for Byzantine nodes (the scheduler message delivery), let~$NT_\Theta$ be the set of executions that do not terminate.
	Moreover, let us define $NT_\Theta^i$ to be the event where the correct (good) node $i$ joins, never leaves and never decides.
	Similarly, for every $n\in\mathbb{N}$ we define $NT_\Theta^{n,i}$ to be the event where the correct (good) node $i$ joins, does not leave, and does not decide within the first $n$ steps of the execution.
\end{definition}
First, note that we can enumerate the correct/good nodes since they are countable.
Second, note that our definition of termination in Section~\ref{sec:model} implies $NT_\Theta = \cup_{i=1}^{\infty}NT_\Theta^i$.
\begin{definition}\label{def:RVs}
	Given a strategy $\Theta$ for Byzantine nodes (the scheduler message delivery) in~GM+ (~SM+),
	and a correct (good) node $i$, we define the random variable $X_\Theta^i$ for each $\eta\in H_{\Theta}$ as follows:
	\begin{displaymath}
		X^i_\Theta(\eta) = \begin{cases}
			1 &\quad \text{If $i$ joins, never leaves, and never decides during $\eta$,}\\
			0 &\quad\text{Otherwise.}
		\end{cases}
	\end{displaymath}
	Furthermore, let us define the random variables $\{X_\Theta^{n, i}\}_{n=1}^{\infty}$ as follows:
	\begin{displaymath}
		X_\Theta^{n, i}(\eta) = \begin{cases}
			1 &\quad \text{If $i$ joins, does not leave, and does not decide within the first $n$ steps during $\eta$,}\\
			0 &\quad\text{Otherwise.}
		\end{cases}
	\end{displaymath}
\end{definition}
\begin{lemma}\label{lem:Xn}
	For every strategy $\Theta$, every correct/good node~$i\in\mathbb{N}$, and every $\eta\in H_\Theta$, we have $\lim_{n\rightarrow\infty}X_\Theta^{n, i}(\eta) = X_\Theta^i(\eta)$, \emph{i.e.}, $X_\Theta^{n, i}$ converges (almost) surely to $X_\Theta^i$.
	%\begin{enumerate}
	%\item $P(E_\Theta^i) = E\{X_\Theta^i\}$.
	%\item $P(E_\Theta^{n,i}) = E\{X_\Theta^{n,i}\}$.
	%\item For every     $\eta\in\Omega_\Theta$ we have $\lim_{n\rightarrow\infty}X_\Theta^{n, i}(\eta) = X_\Theta^i(\eta)$.
	%Thus, $X_\Theta^{n, i}$ converges (almost) surely to $X_\Theta^i$.
	%\item $X_\Theta^{n, i}\leq 1$.
	%\end{enumerate}
\end{lemma}
\begin{proof}
	Consider any execution $\eta\in H_\Theta$.
	%If $i$ never joins during $\eta$, then we have $X_\Theta^i(\eta) = X_\Theta^{n,i}(\eta) = 0$ for every $n$.
	If node~$i$ joins, never leaves, and never decides during $\eta$, then for every $n$ we have $X_\Theta^i(\eta) = X_\Theta^{n, i}(\eta) = 1$.
 
        If node~$i$ joins and leaves without deciding, then there exists a step~$n_i\in\mathbb{N}$ in which the node leaves.
        Therefore, for every~$n\geq n_i$, we have~$X_\Theta^i(\eta) = X_\Theta^{n, i}(\eta) = 0$.
	If node~$i$ joins and decides during $\eta$, then there exists an $n_i\in\mathbb{N}$ such that node~$i$ decides at step $n_i$.
	Therefore, for every $n\geq n_i$ we have $X_\Theta^i(\eta) = X_\Theta^{n, i}(\eta) = 0$.
	
\end{proof}
Based on Lemma~\ref{lem:Xn}, and the dominated convergence theorem~\cite{jacod-protter}, we have the following lemma.
\begin{lemma}\label{lem:convergence}
	For every strategy $\Theta$ and every node $i$ we have $\lim_{n\rightarrow\infty}E\{X_\Theta^{n, i}\} = E\{X_\Theta^i\}$.
\end{lemma}
\begin{proof}
	Based on Lemma~\ref{lem:Xn}, the sequence of random variables $\{X_\Theta^{n, i}\}$ converges pointwise to $X_\Theta^i$.
	Furthermore, it is clear that every $X_\Theta^{n, i}$ is non-negative and bounded from above by~1.
	This satisfies the conditions required for the dominated convergence theorem, thus the theorem tells us~$\lim_{n\rightarrow\infty}E\{X_\Theta^{n, i}\} = E\{X_\Theta^i\}$.
\end{proof}
\begin{lemma}\label{lemma:strategy-interpret-appendix}
    For every strategy~$\Psi$ in~GM+, there exists a strategy~$\Lambda$ in~SM+ such that~$H_\Lambda = \{\textsc{Interpret}(\eta)|\eta\in H_\Psi\}$.
\end{lemma}
\begin{proof}
    We show that the strategy~$\Lambda$ exists, and is in fact the same as~$\Psi$.
    For brevity, let~$I_\Psi$ denote~$\textsc{Interpret}(H_\Psi)$, and consider any execution~$\eta$ in~$H_\Psi$.
    According to Lemma~\ref{lemma:interpret}, correct nodes in~$\eta$ receive the same messages, at the same steps, as the good nodes in~$\textsc{Interpret}(\eta)$.
    Thus, the message history up to any step~$s$ in~$\eta$ is the same as the message history up to the same step in~$\textsc{Interpret}(\eta)$.
    In addition, since~$\textsc{Interpret}(\eta)$ is an actual~SM+ execution, compatibility is trivially satisfied.
    Thus, we conclude that the executions in~$I_\Psi$ follow the same strategy as in~$\Psi$.
    Naming this strategy~$\Lambda$ finishes the proof.
    
\end{proof}
% \begin{lemma}
% 	If there exists a strategy $\Psi$ for Byzantine nodes in~GM+ that achieves a non-zero non-termination probability,
% 	then there exists a scheduler strategey~$\Lambda$ in~SM+ that also achieves a non-zero non-termination probability.
% \end{lemma}
\lemmainterpretstrategy*
\begin{proof}
	Let us assume that there exists an environment~$\mathcal{E}$ and a Byzantine strategy $\Theta_{\mathcal{E}}$ in the Gorilla protocol that        leads to a positive non-termination probability.
	According to Lemma~\ref{lemma:reorg-strategy}, there exists a strategy $\Psi$ in~GM+
	that achieves a positive non-termination probability, in environment~$\mathcal{E}$.
        Consider the~$\textsc{Interpret}$ mapping.
        Since, according to Lemma~\ref{lemma:interpret}, the joining/leaving and initial value schedules for correct nodes are bijectively mapped by the~$\textsc{Interpret}$ mapping to the joining/leaving and initial value schedules of the good nodes, respectively, we just set~$\mathcal{E}' = \mathcal{E}$.
        In the rest of the proof, we omit the environments for brevity.

        First, note that according to Lemma~\ref{lemma:decision-interpret}, the~$\textsc{Interpret}$ mapping preserves all of the coin tosses in~$H_\Psi$.
        Moreover, for a given execution~$\eta\in H_\Psi$, the same lemma tells us that~$\textsc{Interpret}(\eta)$ might include \emph{more} coin tosses than those in~$\eta$.
	
	Consider $\Lambda = \textsc{Interpret}(\Psi)$, based on Lemma~\ref{lemma:strategy-interpret-appendix}.	
	Let us also define the events $NT_\Psi, NT_\Psi^i$, $NT_\Psi^{n,i}$, $NT_\Lambda, NT_\Lambda^i$, and $NT_\Lambda^{n,i}$
	based on Definition~\ref{def:events}, for the strategies $\Psi$ and $\Lambda$.
	Since $P_{H_{\Psi}}(NT_\Psi) = P_{H_{\Psi}}(\cup_{i=1}^{\infty}NT_\Psi^i)$, using the union bound we have $P_{H_{\Psi}}(NT_\Psi)\leq \sum_{i=1}^{\infty}P_{H_{\Psi}}(NT_\Psi^i)$.
	Now, $\Psi$ achieves a positive non-termination probability, \emph{i.e.}, $P_{H_{\Psi}}(NT_\Psi) > 0$.
	Therefore, there should exist some $i^*\in\mathbb{N}$ such that $P_{H_{\Psi}}(NT_\Psi^{i^*})>0$, since otherwise we would have $P_{H_{\Psi}}(NT_\Psi)=0$ based on the union bound.
	We now define the random variables $X_\Psi^{i^*}$ and $X_\Psi^{n,i^*}$ as in Definition~\ref{def:RVs}, and based on Lemma~\ref{lem:convergence} we have $\lim_{n\rightarrow\infty}E\{X_\Psi^{n, i^*}\} = E\{X_\Psi^{i^*}\}$.
	Similarly, we define the random variables $X_\Lambda^{i^*}$ and $X_\Lambda^{n,i^*}$, and once again Lemma~\ref{lem:convergence}
	tells us $\lim_{n\rightarrow\infty}E\{X_\Lambda^{n, i^*}\} = E\{X_\Lambda^{i^*}\}$.

	Based on Definition~\ref{def:events}, we know that for every~$n$, we have~$NT_\Psi^{n, i^*}\subset NT_\Psi^{n+1, i^*}$ and~$NT_\Lambda^{n, i^*}\subset NT_\Lambda^{n+1, i^*}$.
        Therefore, we have~$E\{X_\Psi^{n, i^*}\} = P_{H_{\Psi}}(NT_\Psi^{n, i^*})\leq P_{H_{\Psi}}(NT_\Psi^{n+1, i^*}) = E\{X_\Psi^{n+1, i^*}\}$ and~$E\{X_\Lambda^{n, i}\} = P_{H_{\Lambda}}(NT_\Lambda^{n, i^*})\leq P_{H_{\Lambda}}(NT_\Lambda^{n+1, i^*}) = E\{X_\Lambda^{n+1, i^*}\}$.
        Since the increasing sequence~$\{E\{X_\Psi^{n, i^*}\}\}_{n=1}^{\infty}$ converges to~$E\{X_\Psi^{i^*}\} = P_{H_{\Psi}}(NT_\Psi^{i^*}) > 0$, there should exist a step~$n^*$ such that~$P_{H_{\Psi}}(NT_\Psi^{n^*, i^*}) = E\{X_\Psi^{n^*, i^*}\} > 0$.
        Now, given~$n^*$, let us consider~$P_{H_{\Lambda}}(NT_\Lambda^{n^*, i^*}) = E\{X_\Lambda^{n^*, i^*}\}$.
        This value is computed based on the first~$n^*$ steps of the executions in~$H_\Lambda$.
        Based on Lemma~\ref{lemma:decision-interpret}, we know that these executions contain all of the coin tosses happening in the first~$n^*$ steps of the corresponding executions in~$H_\Psi$.
        Moreover, they might contain more coin tosses as explained above.
        Therefore, if the probability of the event~$NT_\Psi^{n^*, i^*}$ in~$H_\Psi$ is positive, the probability of the corresponding event~$NT_\Lambda^{n^*, i^*}$ in~$H_\Lambda$ should also be positive, \emph{i.e.},~$E\{X_\Lambda^{n^*, i^*}\}>0$.
        Since the sequence~$\{E\{X_\Lambda^{n, i^*}\}\}_{n=1}^{\infty}$ is increasing with~$n$ and converging to~$E\{X_\Lambda^{i^*}\}$, we should therefore have~$E\{X_\Lambda^{i^*}\} > 0$.
        It immediately follows that~$P_{H_{\Lambda}}(NT_\Lambda^{i^*}) = E\{X_\Lambda^{i^*}\} > 0 $.
 
	Finally, since $NT_\Lambda^{i^*}\subset NT_\Lambda$, we should have $0 < P_{H_{\Lambda}}(NT_\Lambda^{i^*})\leq P_{H_{\Lambda}}(NT_\Lambda)$.
        This means that~$\Lambda$ is a scheduler strategy that achieves positive non-termination probability, and finishes our proof.

\end{proof}

\thmtermprobone*
\proofthmtermprobone*